\documentclass[11pt,letterpaper]{article}

\usepackage[letterpaper,margin=1.0in]{geometry}
\usepackage{url}
\usepackage{amssymb}
\usepackage{graphicx}
\usepackage{amsmath}
\usepackage{amsfonts}
\usepackage{amssymb}
\usepackage{amsthm}
\usepackage{algorithm}
\usepackage{algorithmic}

\usepackage[pdftex,
            colorlinks,
            plainpages=false,
            hyperindex,
            bookmarksopen,
            linkcolor=black,
            citecolor=black,
            urlcolor=blue]{hyperref}
\usepackage[labelfont=bf]{subcaption}
\usepackage{multirow}
\usepackage{xcolor}
\usepackage{float}

\usepackage{authblk}

\usepackage[normalem]{ulem}

\newtheorem{thm}{Theorem}[section]

\newtheorem{defin}{Definition}[section]

\title{Noise-Adaptive Quantum Compilation Strategies Evaluated with Application-Motivated Benchmarks}

\author[1,2]{Davide Ferrari}
\author[1,2]{Michele Amoretti}
\affil[1]{Department of Engineering and Architecture - University of Parma, Italy}
\affil[2]{Quantum Information Science @ University of Parma, Italy}

\date{}

\begin{document}

\maketitle

\begin{abstract}

Quantum compilation is the problem of translating an input quantum circuit into the most efficient equivalent of itself, taking into account the characteristics of the device that will execute the computation. Compilation strategies are composed of sequential passes that perform placement, routing and optimization tasks. Noise-adaptive compilers do take the noise statistics of the device into account, for some or all passes. The noise statics can be obtained from calibration data, and updated after each device calibration.

In this paper, we propose a novel noise-adaptive compilation strategy that is computationally efficient. The proposed strategy assumes that the quantum device coupling map uses a heavy-hexagon lattice. Moreover, we present the application-motivated benchmarking of the proposed noise-adaptive compilation strategy, compared with some of the most advanced state-of-art approaches. The presented results seem to indicate that our compilation strategy is particularly effective for deep circuits and for square circuits.

\textbf{keywords} - \textit{Noise-adaptive quantum compilation; Quantum programming tools; Quantum benchmarking}
\end{abstract}

\section{Introduction}

Current quantum computers are noisy, characterized by a reduced number of qubits (5-50) with non-uniform quality and highly constrained connectivity. These systems and near-term ones (with 1000 qubits or less), which are denoted as Noisy Intermediate-Scale Quantum (NISQ), may be able to surpass the capabilities of today's most powerful classical digital computers, for some problems. However, noise in quantum gates limits the size of quantum circuits that can be executed reliably. To reliably process information, the physical qubits should maintain the quantum state for sufficiently long. The qubits should also support sufficiently precise operations to allow for correct state manipulation during the coherence window. Last but not least, the qubits should support accurate measurement operations. 

Quantum compilation is the problem of translating an input quantum circuit into the most efficient equivalent of itself~\cite{Corcoles2020}, taking into account the characteristics of the device that will execute the computation. In general, the quantum compilation problem is NP-Hard~\cite{Botea2018,Soeken2019}.

Compilation strategies are composed of sequential \textit{passes} that perform placement, routing and optimization tasks. Placement is performed once, at the very beginning of the compiling process. It is the task of defining a mapping between the virtual qubits of the input quantum circuit and the physical qubits of the device. Routing is the task of modifying the circuit in order to move through subsequent mappings, by means of a clever swapping strategy, in order to conform to the qubit layout of the device. Optimization is the task of minimizing some property of the circuit in order to reduce the impact of noise. Several routing and optimization passes may be executed within a compilation strategy. 

Noise-adaptive compilers do take the noise statistics of the device into account~\cite{Murali2019,Niu2020, Nishio2020, Sivarajah2020}, for some or all passes. The noise statics can be obtained from calibration data, and updated after each device calibration. For example, Qiskit\footnote{\url{https://qiskit.org/}} allows for programmatic retrieval of IBM Q devices' calibration data.

In the literature, compiled circuits are frequently evaluated in terms of depth and gate count overhead with respect to the input circuits. Calculating these figures of merit does not require to execute the compiled circuit. Another common method is to run the compiled circuit many times with a figure of merit being the success rate, i.e., the fraction of runs that resulted in a correct (classical) answer. The Hellinger fidelity \cite{Pollard2001} is a measure of the distance between two random distribution. It is frequently  used to compare the sampled distribution of the results of a quantum computation to the theoretical distribution (provided that the latter one is known).

Recently, Mills et al.~\cite{Mills2020} have presented a framework for  application-motivated benchmarking of full quantum computing stacks. The benchmarks defined there have a \textit{circuit class}, describing the type of circuit to be run on the system, and a \textit{figure of merit}, quantifying how well the system did when running circuits from that class. The idea is that a circuit class should represent a particular application domain. The application-motivated circuit classes proposed by Mills et al. draw inspiration from quantum algorithmic primitives \cite{Blume2019} and from the literature on near-term quantum computing applications (e.g., machine learning and chemistry).

\subsection{Our Contributions}

The first contribution of this paper is a novel noise-adaptive compilation strategy that is computationally efficient. The proposed strategy assumes that the coupling map, i.e., the device-specific directed graph whose vertices correspond to the physical qubits and edges correspond to permitted CNOT gates, uses a heavy-hexagon lattice. In Section~\ref{sec:strategy}, we describe the proposed strategy and we thoroughly analyze its theoretical performance.  

The second contribution is the application-motivated benchmarking of the proposed noise-adaptive compilation strategy, compared with some of the most advanced state-of-art approaches. In Section~\ref{sec:benchmarks}, we summarize the features of three notable circuit classes (deep, square and shallow) and we recall the definitions of five significant figures of merit. Then, in Section~\ref{sec:eval}, we show and discuss the evaluation results.

\section{Related Work}
Despite the quantum compilation problem has been studied for years, there are still few noise-adaptive compiling strategies. 

Qiskit's \textsf{NoiseAdaptiveLayout} placement pass associates a physical qubit to each virtual qubit of the circuit using calibration data, based on the heuristic method proposed by Murali et al.~\cite{Murali2019}. The pass maps virtual qubit pairs in order of decreasing frequency of the CNOT occurrences between them. If a pair exists with both qubits unmapped, the pass picks the best available physical qubit pair, based on CNOT reliability, to map it. If a pair has only one qubit unmapped, the pass maps that qubit to a location that ensures maximum reliability for CNOTs with previously mapped qubits. In the end if there are unmapped qubits, the pass maps them to any available physical qubit.

Nishio et al.~\cite{Nishio2020} proposed a placement pass and a routing pass. The placement pass, which is denoted as \textsf{Greatest Connecting Edge Mapping}, leverages a strategy that is very similar to the one proposed by Murali et al.~\cite{Murali2019}. The routing pass is a beam search algorithm with an heuristic cost function based on the estimated success probability of candidate SWAP gates.

t$|$ket$\rangle$'s \textsf{NoiseAwarePlacement}~\cite{Sivarajah2020} searches for candidate partial placements of virtual qubits to the physical ones. This is done by casting the problem as finding a subgraph monomorphism between the coupling map and a graph representing virtual qubit CNOT interactions in the circuit. Where different possible candidates placement are found, using a heuristic approach, the pass chooses the one with the maximum expected overall fidelity.

Recently, Niu et al.~\cite{Niu2020} have implemented a hardware-aware routing pass (\textsf{HA}), inspired by the work of Li et al.~\cite{Li2019}, that iteratively selects the best scoring SWAP with respect to calibration data as well as qubits distance. The influence of each of these factors on the cost function can be set by means of tunable weights.

\section{Proposed Compilation Strategy}
\label{sec:strategy}

We introduce an heuristic-based compilation strategy consisting of two compilation passes that account for calibration data such as gate reliability and readout errors. The first one is a placement pass, i.e., a pass that finds an initial mapping of virtual qubits to physical ones. The second pass is a routing strategy to make all two-qubits gates compliant with hardware connectivity. Both passes exploit device calibration data with the aim of improving the quality of the states produced by the compiled circuit. In the following we assume a heavy-hexagon lattice for the coupling map of the device, such as the one used by IBM superconducting devices \cite{IBMQ}. In heavy-hexagon lattices, the qubits are located on the nodes and edges of each hexagon. Each qubit has either two or three neighbors, meaning the graph has vertices of degree 2 or 3. As a consequence, only three different frequency assignments are necessary for the superconducting qubits, as opposed to a square lattice, which naturally requires at least five different frequencies for addressability. The heavy-hexagon lattice also greatly reduces crosstalk errors since, in principle, only qubits on the edges of the lattice need to be driven by cross-resonance (CR) drive tones \cite{Zhu2020}.

\subsection{Placement}

In the proposed placement pass, a sequence of qubits connected in a line is detected. On a coupling map with heavy-hexagon connectivity, it not possible to find a line connecting all qubits. However, it is possible to find lines that traverse a good portion of the graph, with leftover qubits that could later be used for swapping across segments of the line. Fig.~\ref{fig:ibmq_manhattan} shows the coupling map of the 65 qubits \textit{ibmq\_manhattan} device, with the sequence of qubits used for the initial mapping highlighted in light-blue. 

\begin{figure}
    \centering
    \includegraphics[width=8cm]{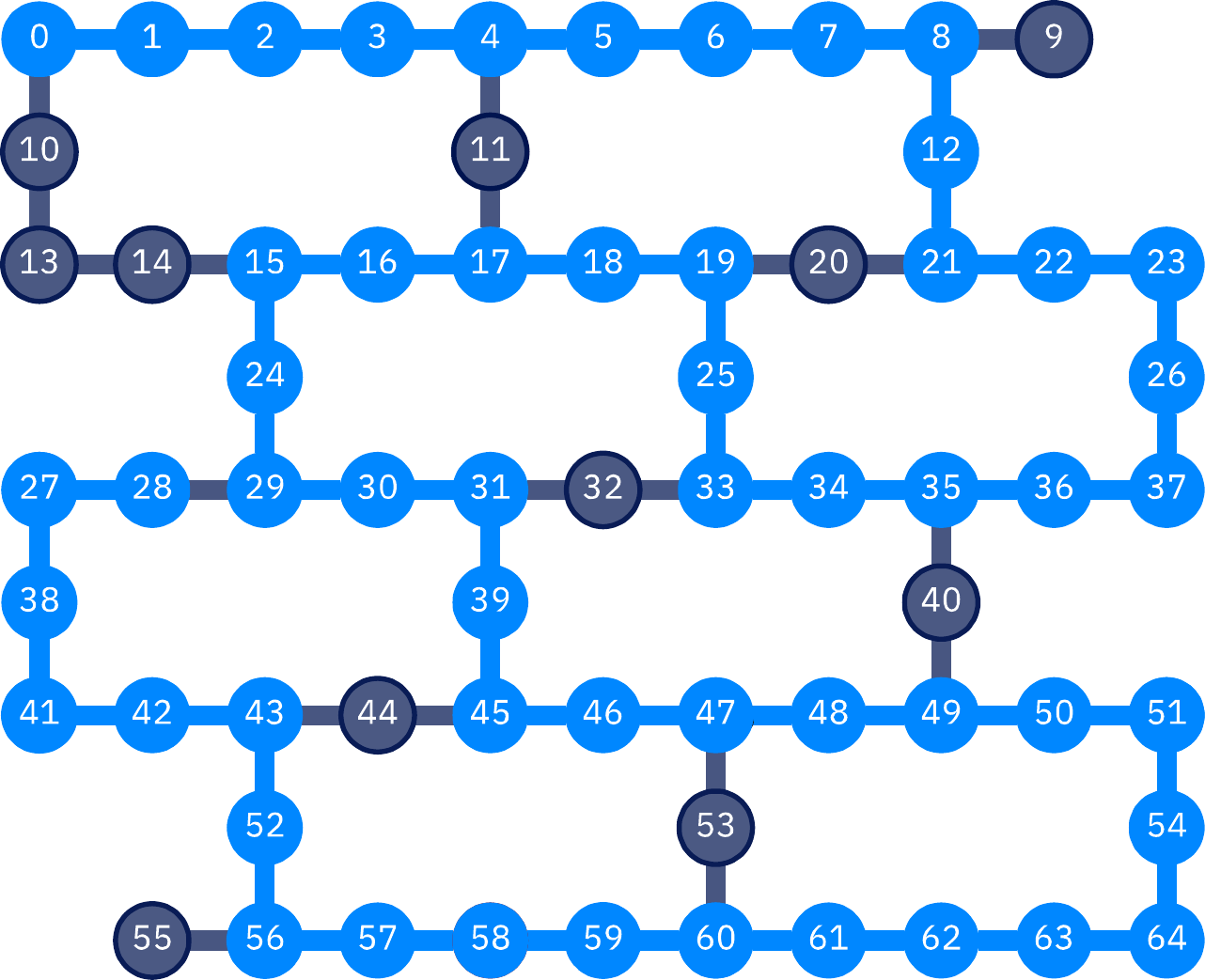}
    \caption{Initial mapping on the 65 qubits of the \textit{ibmq\_manhattan} device, highlighted in light-blue.}
    \label{fig:ibmq_manhattan}
\end{figure}

To find a line of qubits, one can leverage the features of the graph, such as its regular structure and the fact that every qubit is identified by a number in $\{0,...,n-1\}$, where $n$ is the number of qubits in the device. This is done starting from node $0$ and traversing the graph while keeping track of already visited nodes, backtracking if a dead end is reached before having explored the entire graph.

If the found sequence contains more qubits than needed by the circuit, the pass selects a best subset based on two-qubit gates reliability using a sliding window technique. On the other hand, if there are not enough qubits in the sequence, the pass proceeds to insert qubits from the leftover ones, depicted in dark gray in Fig.~\ref{fig:ibmq_manhattan}, until the necessary number of qubits is reached. These qubits are first scored based on calibration data, and then inserted into the chain after one of their neighboring qubits, starting form the one with the highest score.

\subsection{Noise Adaptive Swap}

Given an initial mapping of virtual qubits to physical ones, the \textsf{NoiseAdaptiveSwap} pass proceeds to compute a \textit{front layer} of non hardware-compliant CNOT gates. 

\begin{defin}
Let $C_{k}$ be a quantum circuit on $k$ qubits and $U_{i}(\mathcal{Q}_{i})$ a quantum gate acting on a set of qubits $\mathcal{Q}_{i}$ with $0<|\mathcal{Q}_{i}|\leq k$. A \textsf{layer} $\mathcal{L}$ is a set of consecutive gates that can be applied concurrently such that:
\begin{enumerate}
    \item $\mathcal{Q}_i \cap \mathcal{Q}_j = \emptyset$ for all $U_{i},U_{j} \in \mathcal{L}$, with $i \neq j$
    \item $\sum_{U_{i}\in \mathcal{L}} |\mathcal{Q}_{i}| \leq k$
\end{enumerate}
\end{defin}

\begin{defin}
The \textsf{front layer} is a layer $\mathcal{F}$ such that $|\mathcal{Q}_{i}| = 2$ for all $ U_{i} \in \mathcal{F}$.
\end{defin}

For convenience, we reformulate circuits by means of the Directed Acyclic Graph (DAG) circuit formalism, as it enables to effectively represent gates dependencies in quantum circuits.

\begin{defin}
A \textsf{DAG circuit} is a directed acyclic graph where vertices represent gates and directed edges represent qubit dependencies. A directed edge $e_{q}(i,j)$ between vertices $i$ and $j$ represents a dependency between gate $i$ and gate $j$ with respect to qubit $q$, i.e., gate $i$ must be executed before $j$ and both gates act on qubit $q$.
\end{defin}

\begin{defin}
Given a directed graph $G = (\mathcal{V}, \mathcal{E})$, a \textsf{topological order} of $G$ is a linear ordering over vertices in $\mathcal{V}$ such that, for all directed edges $(v, w) \in \mathcal{E}$, $v$ precedes $w$ in the ordering.
\end{defin}

\begin{defin}
Let $v_{0}v_{1}...v_{n}$ be a topological ordering on graph $G=(\mathcal{V},\mathcal{E})$. Then $v_{j}$ is a \textsf{direct topological successor} of $v_{i}$ iff $j>i$ and $\exists e \in \mathcal{E}$ such that $e=(v_{i},v_{j})$.
\end{defin}

To compute the front layer of non hardware-compliant CNOT gates, the pass iterates over all gates in the circuit, in topological order. Gates that do not need routing, such as one-qubit gates or hardware-compliant CNOT gates are added to the set of \textit{executed} gates $\mathcal{X}$. CNOT gates that need to be properly mapped and do not interact with qubits already interested by a CNOT gate in the front layer, are added to the latter. Every gate that is a successor of a CNOT gate in the front layer, is added to the set of \textit{not executed} gates $\mathcal{\bar{X}}$.

From the front layer, the pass computes a list of possible SWAP operations involving at least one qubit interested by a CNOT in the front layer. These SWAP operations are scored with an heuristic cost function $h(s)$, where $s$ denotes the considered SWAP gate. The cost function is computed over all gates in the \textit{front layer} $\mathcal{F}$ plus a set of upcoming gates $\mathcal{U} \subset \mathcal{\bar{X}}$, as shown in Eq.~\ref{eq:h}, where $\pi_{s}(g_{c})$ and $\pi_{s}(g_{t})$ are the physical qubits corresponding to the control and target of gate $g$ with mapping $\pi_{s}$.

\begin{equation}
    h(s) =  \frac{\alpha}{|\mathcal{F}\cup \mathcal{U}|} \sum_{g \in \mathcal{F}\cup \mathcal{U}} R(\pi_{s}(g_{c}),\pi_{s}(g_{t})) + \frac{1-\alpha}{|\mathcal{F}\cup \mathcal{U}|} \sum_{g \in \mathcal{F}\cup \mathcal{U}} 1-D(\pi_{s}(g_{c}),\pi_{s}(g_{t}))\label{eq:h}
\end{equation}

Here, $R$ and $D$ are respectively the swap paths reliability matrix and the distance matrix for every qubits pair in the coupling map. Matrix $R$ stores in entry $(i,j)$ the reliability of the most reliable swap path between qubits $i$ and $j$, where the reliability of a single SWAP along an edge of the coupling map is computed with regard to CNOT gate and readout error rates. Matrix $D$ is the distance matrix and stores at entry $(i,j)$ the shortest distance between qubits $i$ and $j$. The swap path reliability and the distance between qubits are important components that one would like to maximize and minimize, respectively. The coefficient $\alpha$ gives the opportunity to set equal or opposite weights for them, in the context of the heuristic cost function $h(s)$.

\begin{defin}
The \textsf{reliability} of SWAP $s$ between qubits $i$ and $j$ is $r(s) = \mu(i,j)^3$, where $\mu(i,j)$ is the success rate of a $CNOT(i,j)$ between qubits $i$ and $j$, obtained from calibration data.
\end{defin}

In the previous definition, it is assumed that the SWAP gate is composed of 3 CNOT gates.

\begin{defin}
Let $S = s_{0}s_{1}...s_{k}$ be a sequence of SWAP gates (SWAP path) and denote the reliability of SWAP gate $s$ as $r(s)$. Then the reliability of $S$ is given by
$\prod_{i=0}^{k} r(s_{i})$.
\end{defin}

\begin{defin}
Given a quantum device with $n$ qubits, the SWAP paths \textsf{reliability matrix} $R \in \mathbb{R}^{n \times n}$, is the real matrix that stores in entry $(i,j)$ the maximum reliability or, equivalently, success rate at which qubit $i$ can be moved to a neighbor of qubit $j$ through a sequence of SWAP gates.
\end{defin}

\begin{defin}
Given a quantum device with $n$ qubits, the \textsf{distance matrix} $D \in \mathbb{N}^{n \times n}$, is the matrix that stores in entry $(i,j)$ the minimum distance between qubit $i$ and a neighbor of qubit $j$.
\end{defin}

Both matrices can be efficiently precomputed using the Floyd–Warshall algorithm~\cite{Floyd1962}. Matrix $R$ is computed with edge weights corresponding to SWAP reliability, while matrix $D$ has edge weights equal to 1. The SWAP reliability can be easily derived from CNOT calibration data, as each SWAP is usually achieved with three consecutive CNOTs. As matrix $D$ contains entries with incompatible scales with respect to $R$, both matrices in Eq.~\ref{eq:h} are normalized.

The first part of Eq.~\ref{eq:h} sums over the swap reliability of all gates involved in $\mathcal{F}\cup \mathcal{U}$ and divides by $|\mathcal{F}\cup \mathcal{U}|$ to obtain a mean value. The second part follows a similar approach with the distance matrix, except that the sum must be over $1-D$, as the goal is to maximize the reliability while minimizing the normalized distance. 

\subsection{Computational Complexity}

\begin{thm}
The execution of the \textsf{Placement} pass takes $O(n)$ time, where $n$ is the number of qubits in the device.
\end{thm}

\begin{proof}
In order to find the initial placement of qubits, the \textsf{Placement} pass needs to iterate over all nodes in the coupling map, taking $O(n)$ time, where $n$ is the number of qubits in the device.
\end{proof}

As previously stated, the $R$ and $D$ matrices are precomputed using the Floyd–Warshall algorithm, whose time complexity is $O(n^{3})$.

\begin{thm}
Assuming a coupling map with heavy-hexagon connectivity, computing the $h(s)$ score of a candidate SWAP gate $s$ takes $O(n^{2})$ time, where $n$ is the number of qubits in the device.
\label{thm:score}
\end{thm}

\begin{proof}
Computing the score of a candidate SWAP gate requires retrieving SWAP paths reliability and distance information for all non-compliant CNOT gates in the front layer, and a few additional ones from subsequent layers (by default 5, which can be treated as a constant factor and therefore ignored). There can be at most $n/2$ CNOT gates in one layer, where $n$ is the number of qubits. Thus, computing the $h(s)$ score takes $O(n)$ time.
Assuming a coupling map with heavy-hexagon connectivity, the number of possible different SWAP gates candidates in a layer is $O(n)$. Consequently, to obtain the most promising SWAP gates for the current layer the score function must be computed $O(n)$ times, taking a total of $O(n^{2})$ time.
\end{proof}

A \textit{beam search} algorithm is used to look up for a new possible mapping through a sequence of SWAP gates, with respect to either the initial mapping or the one from a previous iteration of the routing pass. The beam search algorithm is characterized by beam width $w$ and search depth $k$, which gives an $O(w^{k} n^{2})$ time complexity, taking into account the statement from Theorem~\ref{thm:score}.

\begin{thm}
Assuming a coupling map with heavy-hexagon connectivity, the \textsf{NoiseAdaptiveSwap} pass takes $O(\frac{w^{k}}{k} g n^{2.5})$ time, where $g$ is the number of CNOT gates in the circuit.
\end{thm}

\begin{proof}
Beam search with beam width $w$ and search depth $k$ takes $O(w^{k} n^{2})$ time.
Given a coupling map with heavy-hexagon connectivity, the distance between any qubit pair can be reasonably approximated to $O(\sqrt{n})$. In the worst case scenario, one would need to search $O(\sqrt{n})$ mappings for every CNOT gate in the circuit. This can be relaxed to $O(\sqrt{n}/k)$, as each new mapping search will insert $k$ SWAP gates on average. It follows that the whole compilation process should then take $O(g \frac{\sqrt{n}}{k} w^{k} n^{2}) = O(\frac{w^{k}}{k} g n^{2.5})$, where $g$ is the number of CNOT gates in the circuit.
\end{proof}

The $w^k$ factor is due the fact that the pass does not just pick the best scoring SWAP at each iteration, but instead searches over the $k$ best SWAPs, as the best scoring one may not necessarily lead to the best final solution. We could further assume that $w=k$, giving a time complexity of $O(k^{k-1} g n^{2.5})$.

\subsection{Implementation}
A Python implementation of the proposed noise-adaptive compiler is available on GitHub.\footnote{\url{https://github.com/qis-unipr/noise-adaptive-compiler}} It has been designed as a Qiskit pass \cite{IBMQ}, thus it can be used with any quantum device supported by Qiskit.

\section{Application-Motivated Benchmarks}
\label{sec:benchmarks}
With reference to the recent paper by Mills et al.~\cite{Mills2020}, we consider three circuit classes, namely \textit{deep}, \textit{square} and \textit{shallow}.

Deep circuits are constructed from several layers of Pauli gadgets~\cite{Cowtan2020}, which are quantum circuits that implement an operation corresponding to exponentiating a tensor product of Pauli matrices. For example, in quantum chemistry, deep circuits are used to build UCC trial states used in the variational quantum eigensolver (VQE) \cite{Panagiotis2018}.

Square circuits are random circuits built from two-qubit gates. They provide a benchmark at all layers of the quantum computing stack. Indeed, they have been suggested as a means to demonstrate quantum computational supremacy \cite{Boixo2018}.
Square circuits avoid favoring any device in particular, because they allow two-qubit gates to act between any pair of qubits in the uncompiled circuit.

The class of shallow circuits is a subclass of Instantaneous Quantum Polytime (IQP) circuits~\cite{Shepherd2009}. 
IQP circuits consist of gates diagonal in the Pauli-Z basis, sandwiched between two layers of Hadamard gates acting on all qubits. Shallow circuits are characterized by limited connectivity between the qubits and by a depth that increases slowly with width. Thus, shallow circuits are useful for understanding the performance of a device being utilized for applications whose circuit depth grows less quickly than their qubit requirement.

In Section~\ref{sec:eval}, the compiled circuits are evaluated in terms of a few figures of merit that are described below. Let us denote an $n$ qubit circuit as $C$, the ideal output distribution of $C$ as $p_C$ and the output distribution produced by the compiled implementation of $C$ as $D_C$. 

\begin{itemize}
    \item \textbf{Hellinger fidelity} - The Hellinger fidelity between $D_C$ and $p_C$ is 
    \begin{equation}
    F_C = \left(\sum_{x \in \{0,1\}^n} \sqrt{D_C(x)p_C(x)}\right)^2.
    \end{equation}
    We would like that $F_C = 1$.
    \item \textbf{Heavy Output Generation (HOG)} - An output $z \in \{0,1\}^n$ is heavy for a quantum circuit $C$, if $p_C(z)$ is greater than the median of the set $\{p_C(x) : x \in \{0,1\}^n\}$. The HOG probability of $D_C$, i.e., the probability that samples drawn from from $D_C$ will be heavy outputs in $p_C$, is
    \begin{equation}
        \text{HOG}(D_C,p_C) = \sum_{x \in \{0,1\}^n} D_C(x) \delta_C(x)
    \end{equation}
    where $\delta_c(x) = 1$ if $x$ is heavy for $C$, and $\delta_c(x) = 0$ otherwise. We would like HOG$(D_C,p_C) > 1/2$, as it would help us distinguish between a good implementation of $C$ and an attempt to mimic it by generating random bitstrings. Of course this is true if $p_C$ is sufficiently far from uniform.
    \item \textbf{$l_1$-norm distance} - The $l_1$-norm distance between $D_C$ and $p_C$ is
    \begin{equation}
        l_1(D_C,p_C) = \sum_{x \in \{0,1\}^n} |D_C(x) - p_C(x)|.
    \end{equation}
    We would like that $l_1(D_C,p_C) = 0$.
    \item \textbf{CNOT Count} - In a quantum circuit, the number of CNOT gates is denoted as CNOT count.
    \item \textbf{CNOT Depth} - In a quantum circuit, the number of layers containing CNOT gates is denoted as CNOT depth.
\end{itemize}

In Section~\ref{sec:eval}, as suggested in~\cite{Mills2020}, we approximate $D_C$ using samples obtained by running the compiled implementation of $C$ several times. That is, given samples $\mathcal{S} = \{x_i,..,x_m\}$ from $D_C$, let $\mathcal{S}_x$ be the number of times $x$ appears in $\mathcal{S}$ and define $\widetilde{D}_C(x) = \mathcal{S}_x/m$. 

\section{Evaluation}
\label{sec:eval}

\begin{figure*}[!ht]
\centering
\begin{tabular}{ cc }
    \begin{minipage}{0.45\linewidth}
    	\centering
    	\includegraphics[width=\linewidth]{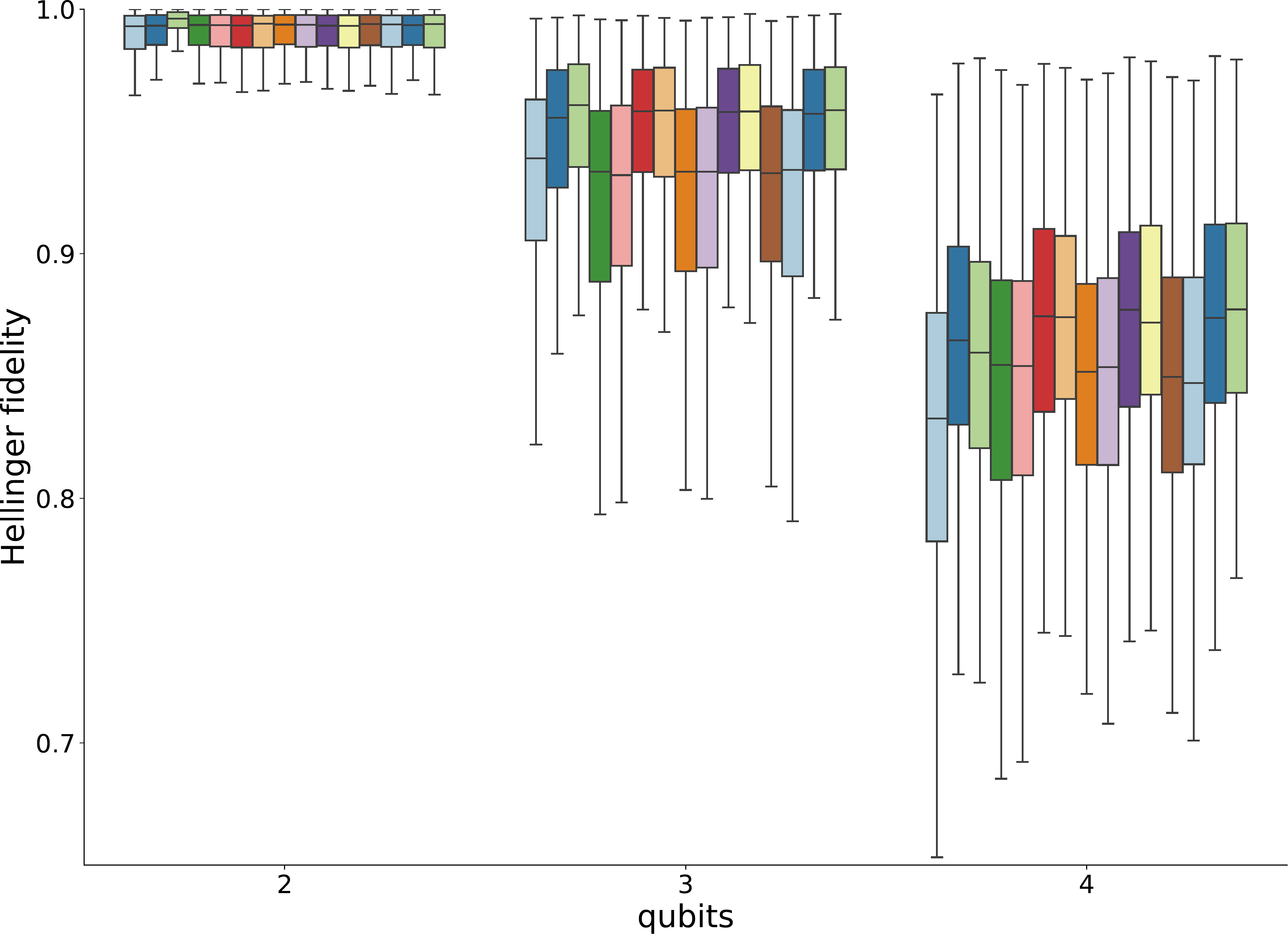}
    	\subcaption{}
        \label{fig:deep_fidelity}
    \end{minipage}
    &
    \begin{minipage}{0.45\linewidth}
    	\centering
    	\includegraphics[width=\linewidth]{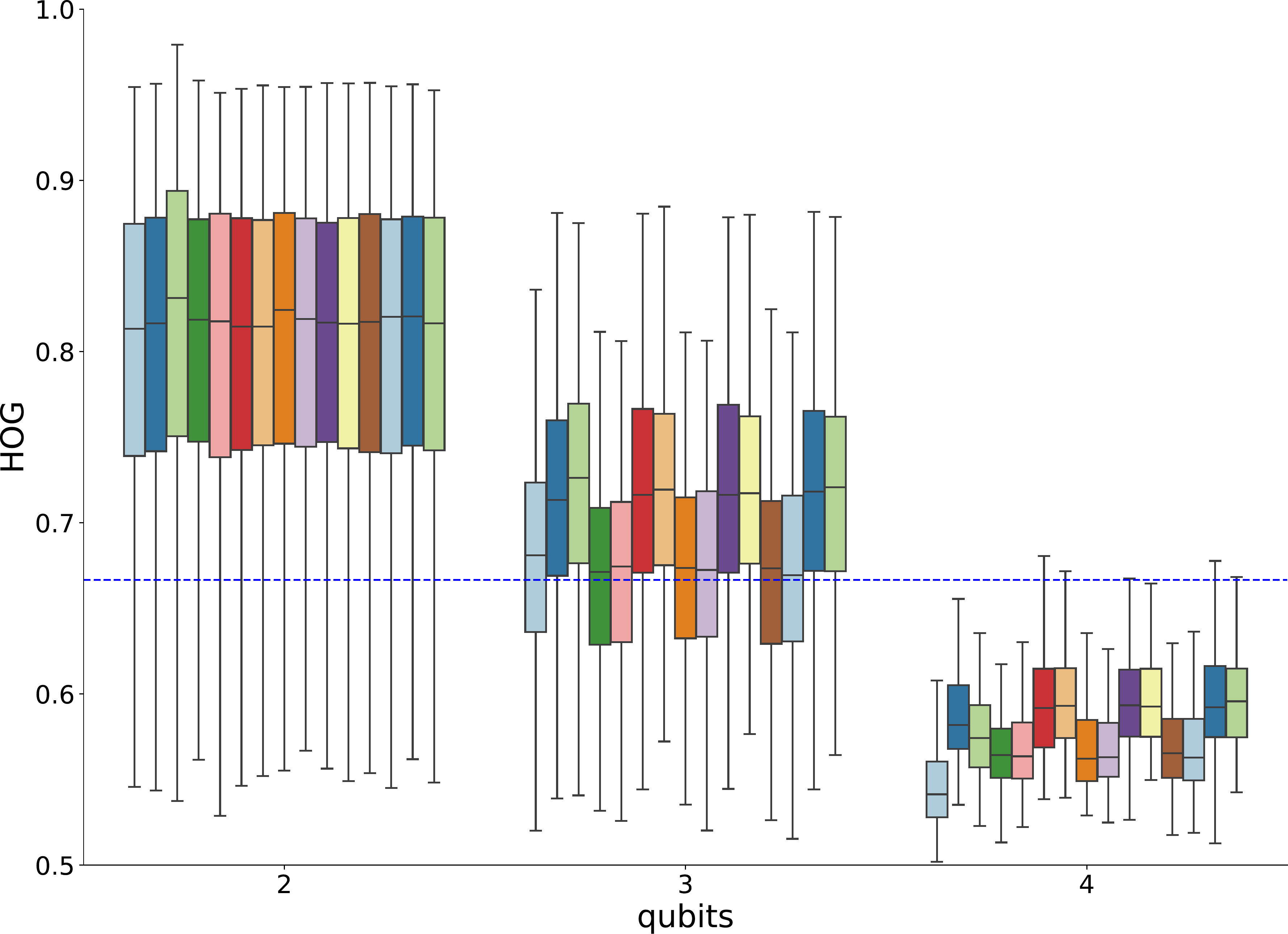}
    	\subcaption{}
        \label{fig:deep_hog}
    \end{minipage}\\
    \begin{minipage}{0.45\linewidth}
		\centering
		\includegraphics[width=\linewidth]{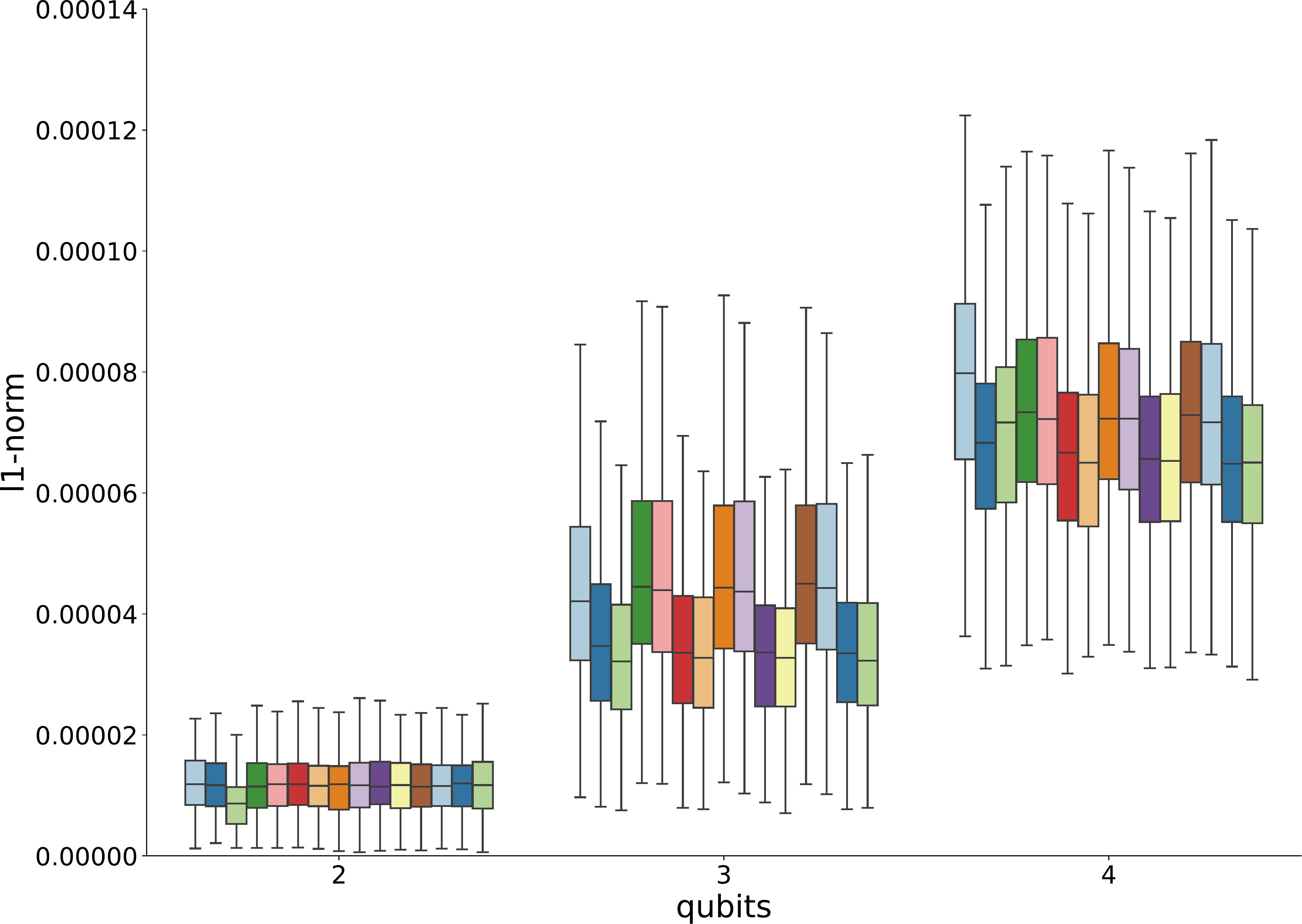}
		\subcaption{}
		\label{fig:deep_l1}
	\end{minipage}
	&
	\begin{minipage}{0.45\linewidth}
    	\centering
    	\includegraphics[width=\linewidth]{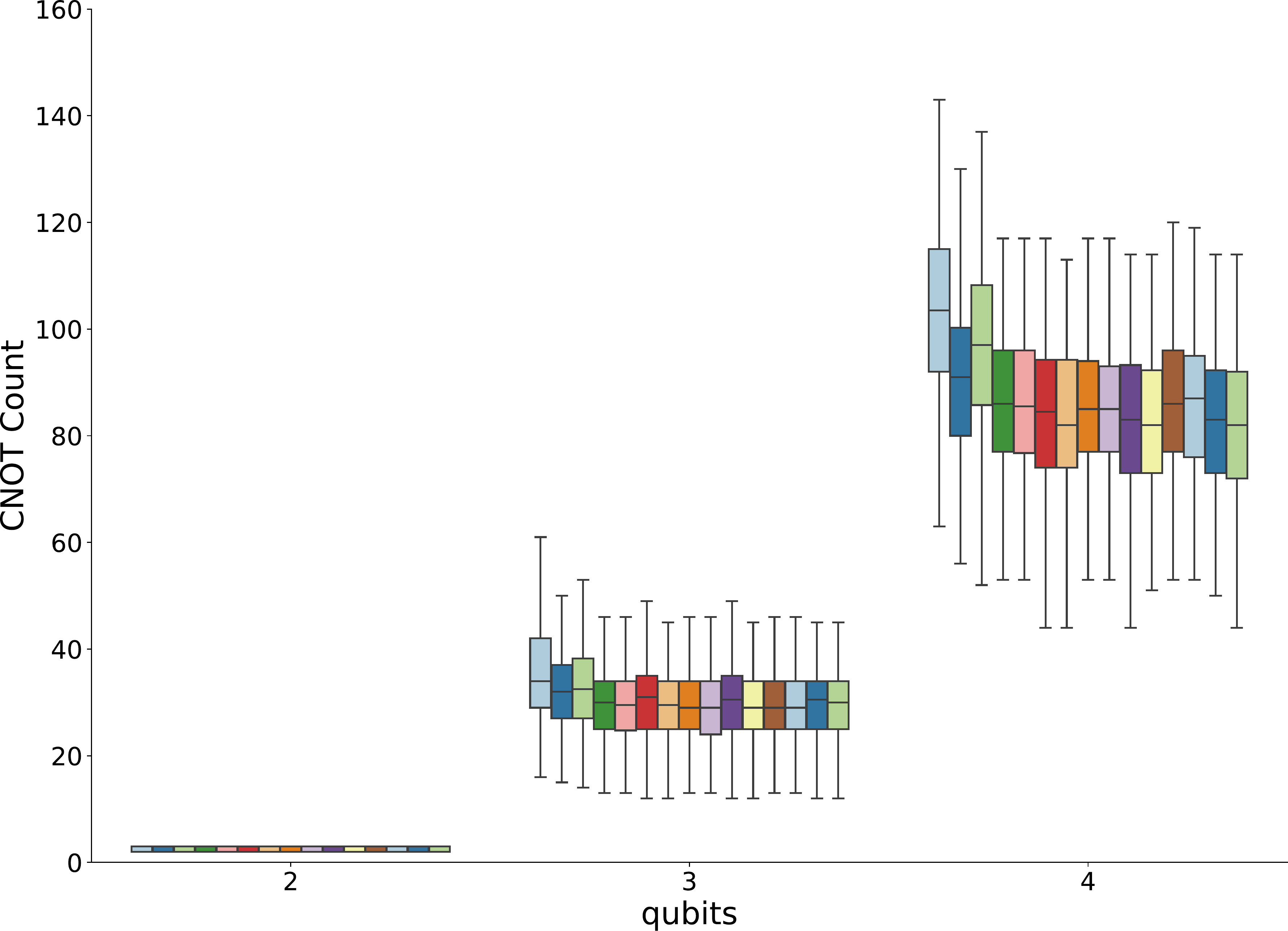}
    	\subcaption{}
        \label{fig:deep_cx_count}
    \end{minipage}\\
    \begin{minipage}{0.45\linewidth}
    	\centering
    	\includegraphics[width=\linewidth]{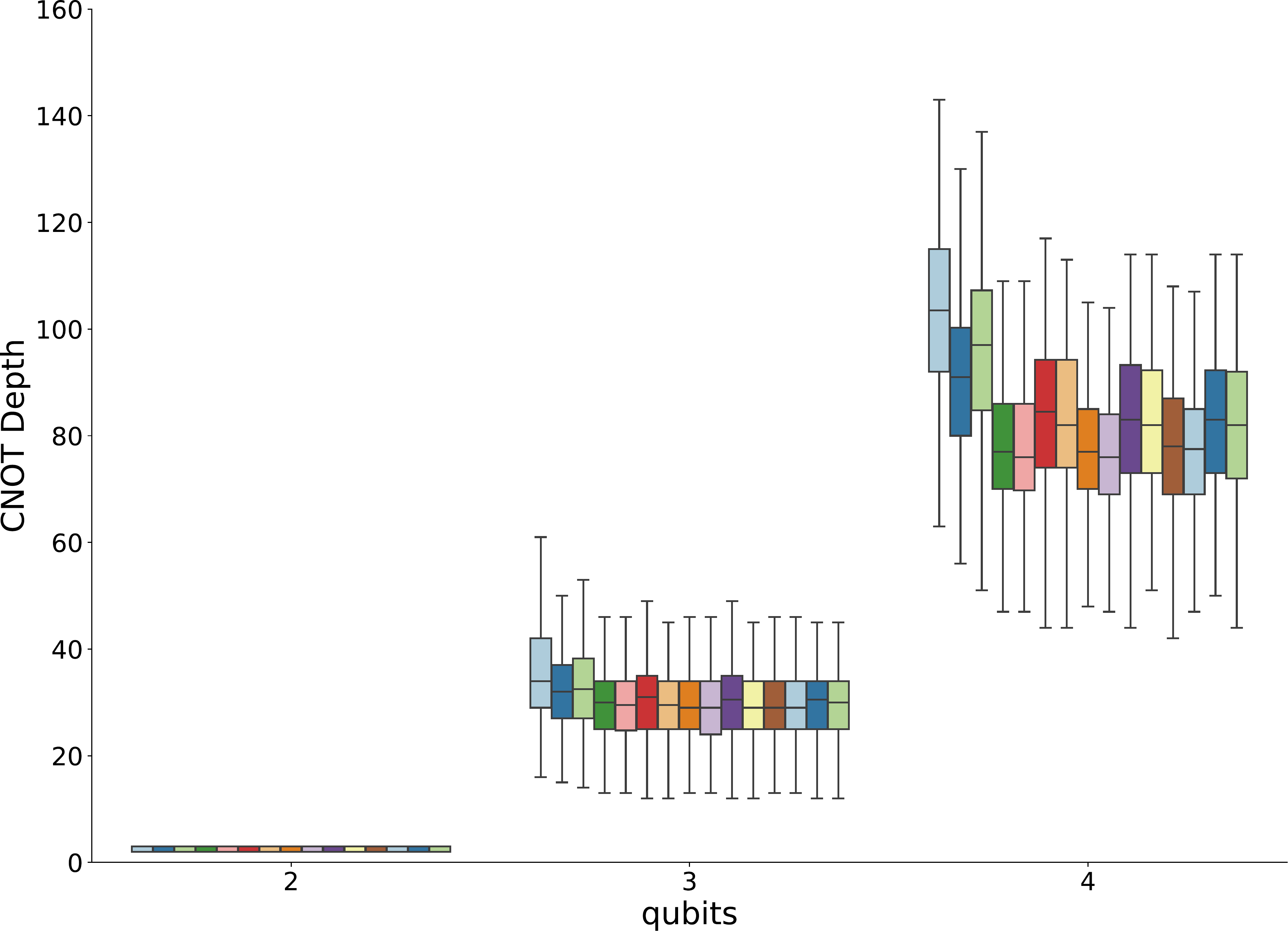}
    	\subcaption{}
        \label{fig:deep_cx_depth}
    \end{minipage}
    &
    \begin{minipage}{7cm}
		\centering
		\includegraphics[width=3cm]{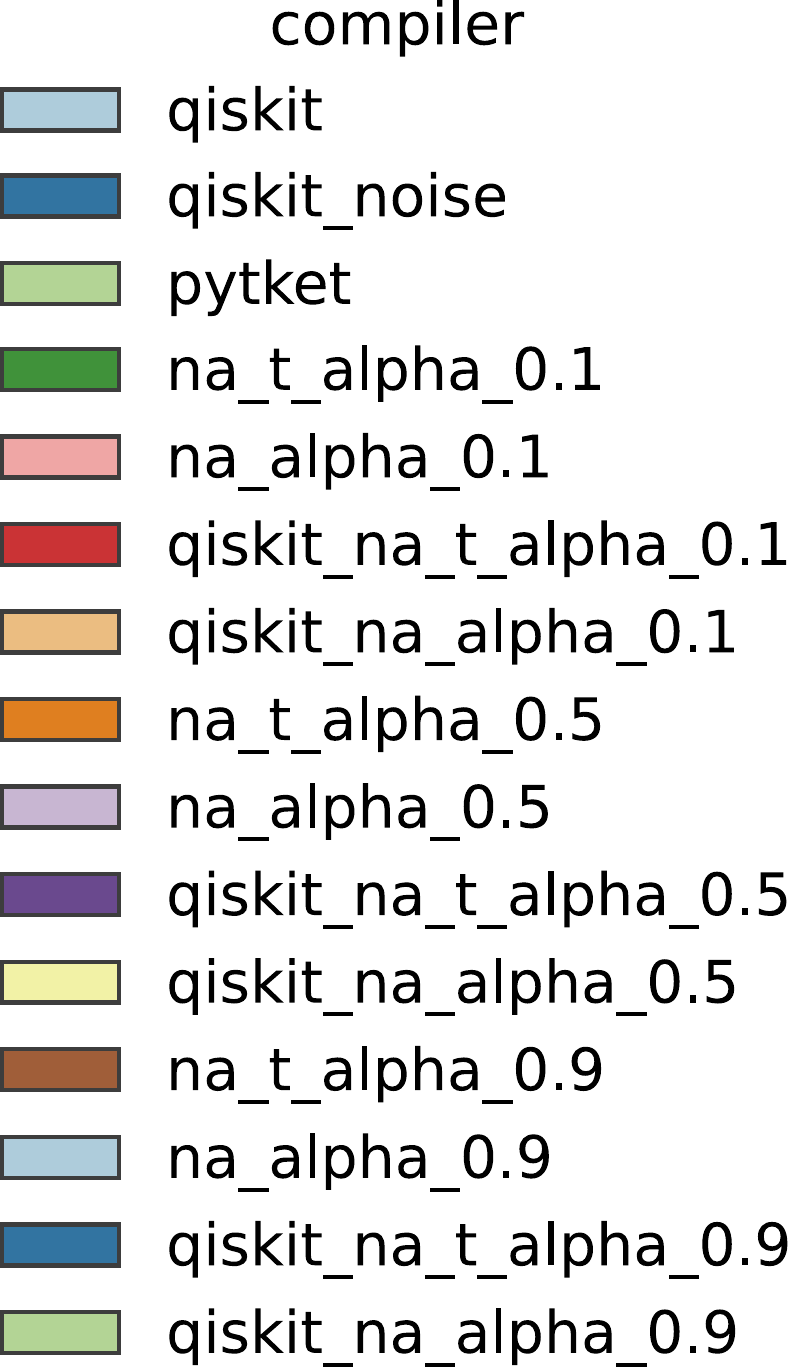}
	\end{minipage}
	\end{tabular}
    \caption{Comparison of compilation strategies with deep circuits using (\ref{fig:deep_fidelity})~Hellinger fidelity, (\ref{fig:deep_hog})~HOG, (\ref{fig:deep_l1})~$l_1$-norm, (\ref{fig:deep_cx_count})~CNOT count and (\ref{fig:deep_cx_depth})~CNOT depth as metrics.  Qiskit’s statevector simulator has been used to sample noisy circuits from the noise model obtained with calibration data of the IBM \textit{ibmq\_casablanca} 7-qubit device.}
    \label{fig:deep}
\end{figure*}

\begin{figure*}[!ht]
\centering
\begin{tabular}{ cc }
    \begin{minipage}{0.45\linewidth}
    	\centering
    	\includegraphics[width=\linewidth]{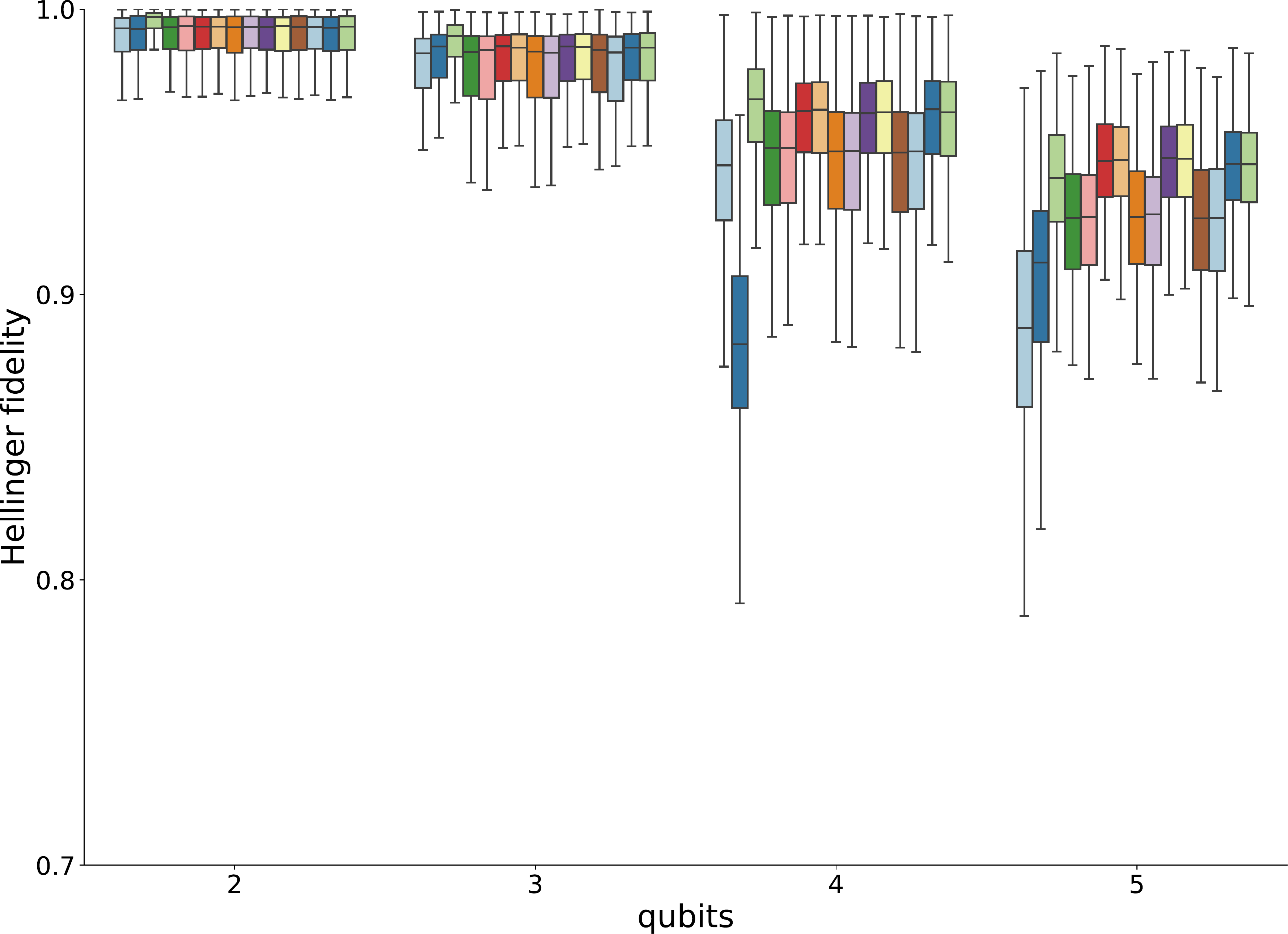}
    	\subcaption{}
        \label{fig:square_fidelity}
    \end{minipage}
    &
    \begin{minipage}{0.45\linewidth}
    	\centering
    	\includegraphics[width=\linewidth]{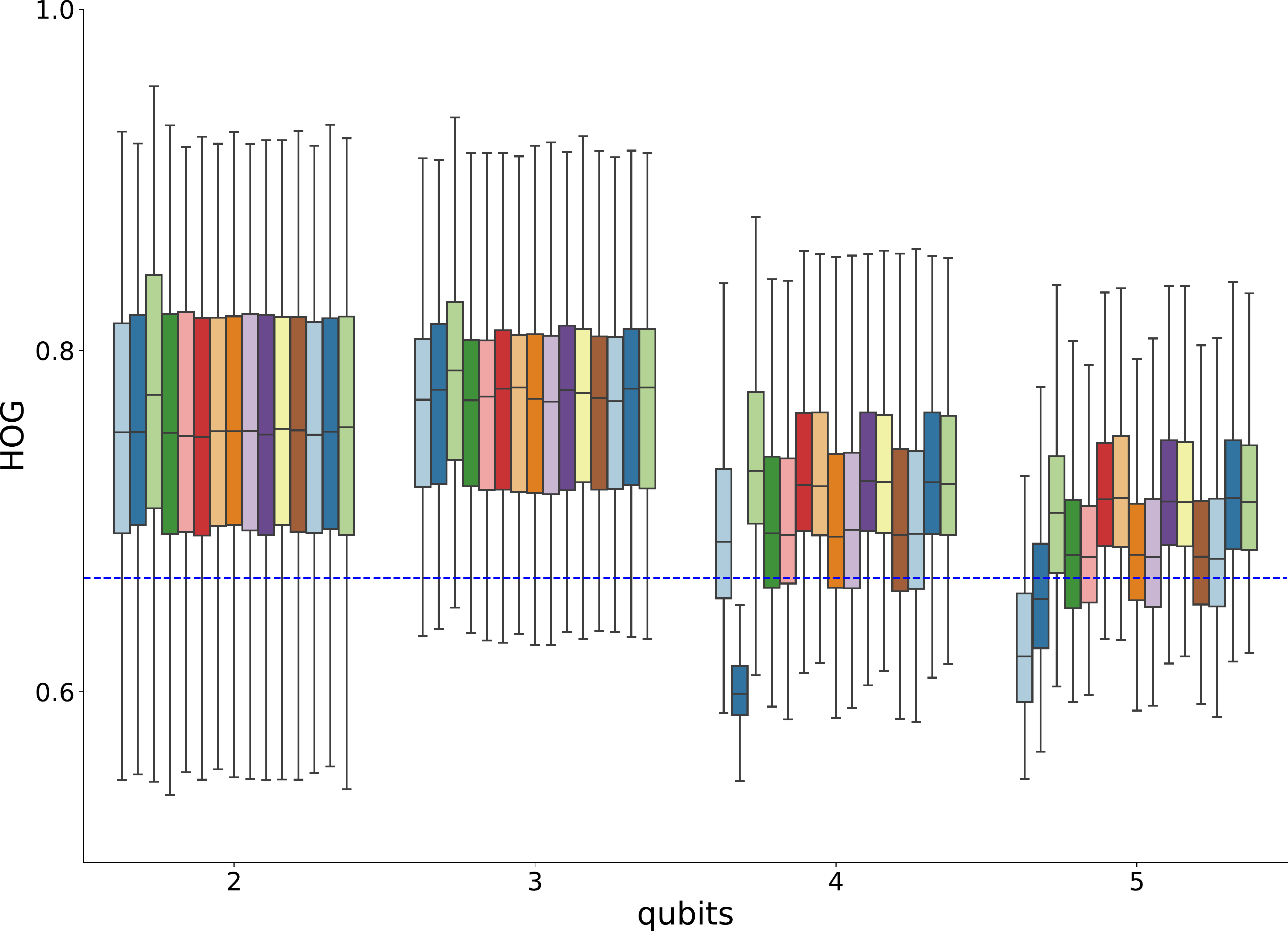}
    	\subcaption{}
        \label{fig:square_hog}
    \end{minipage}\\
    \begin{minipage}{0.45\linewidth}
		\centering
		\includegraphics[width=\linewidth]{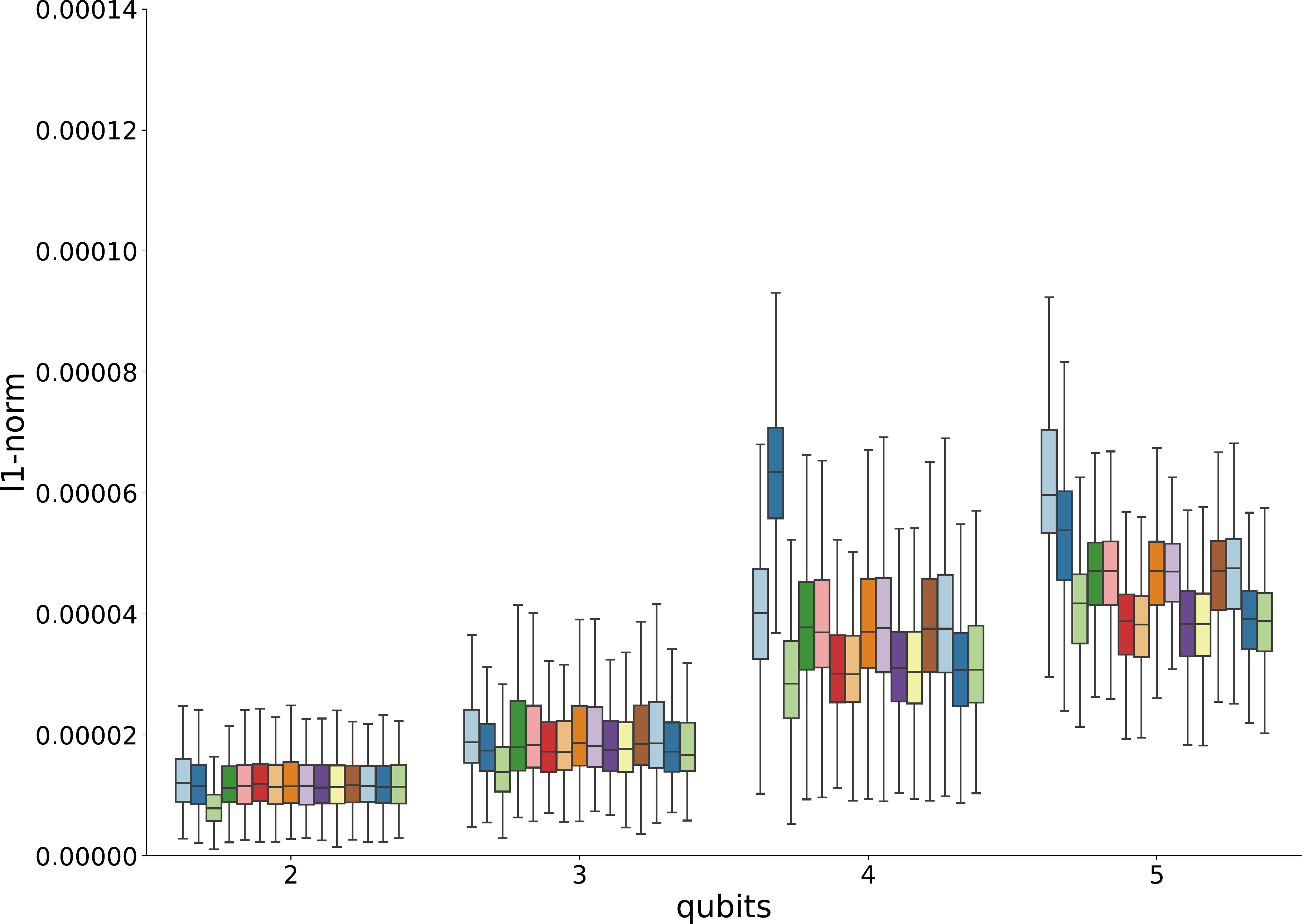}
		\subcaption{}
		\label{fig:square_l1}
	\end{minipage}
	&
	\begin{minipage}{0.45\linewidth}
    	\centering
    	\includegraphics[width=\linewidth]{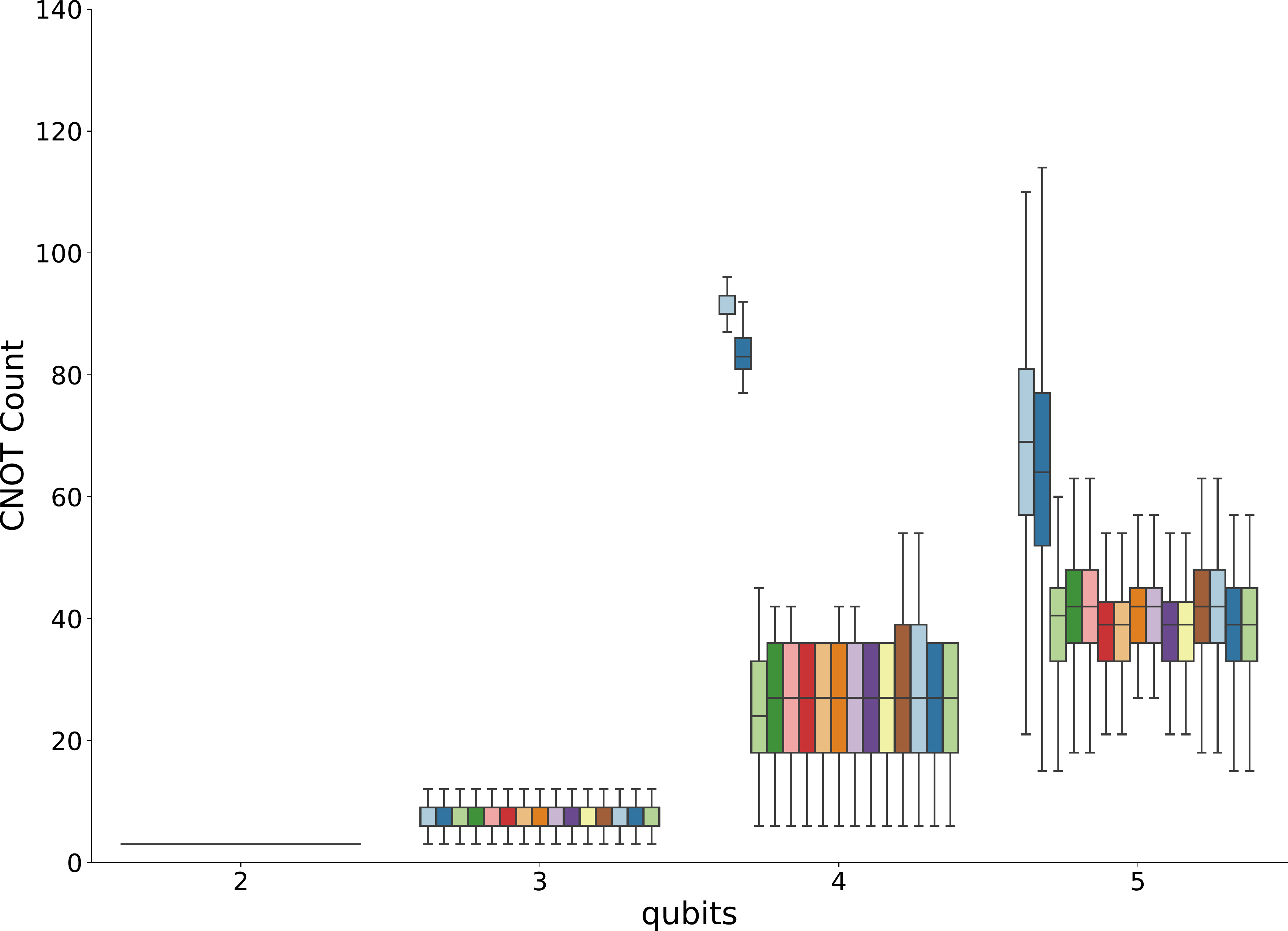}
    	\subcaption{}
        \label{fig:square_cx_count}
    \end{minipage}\\
    \begin{minipage}{0.45\linewidth}
    	\centering
    	\includegraphics[width=\linewidth]{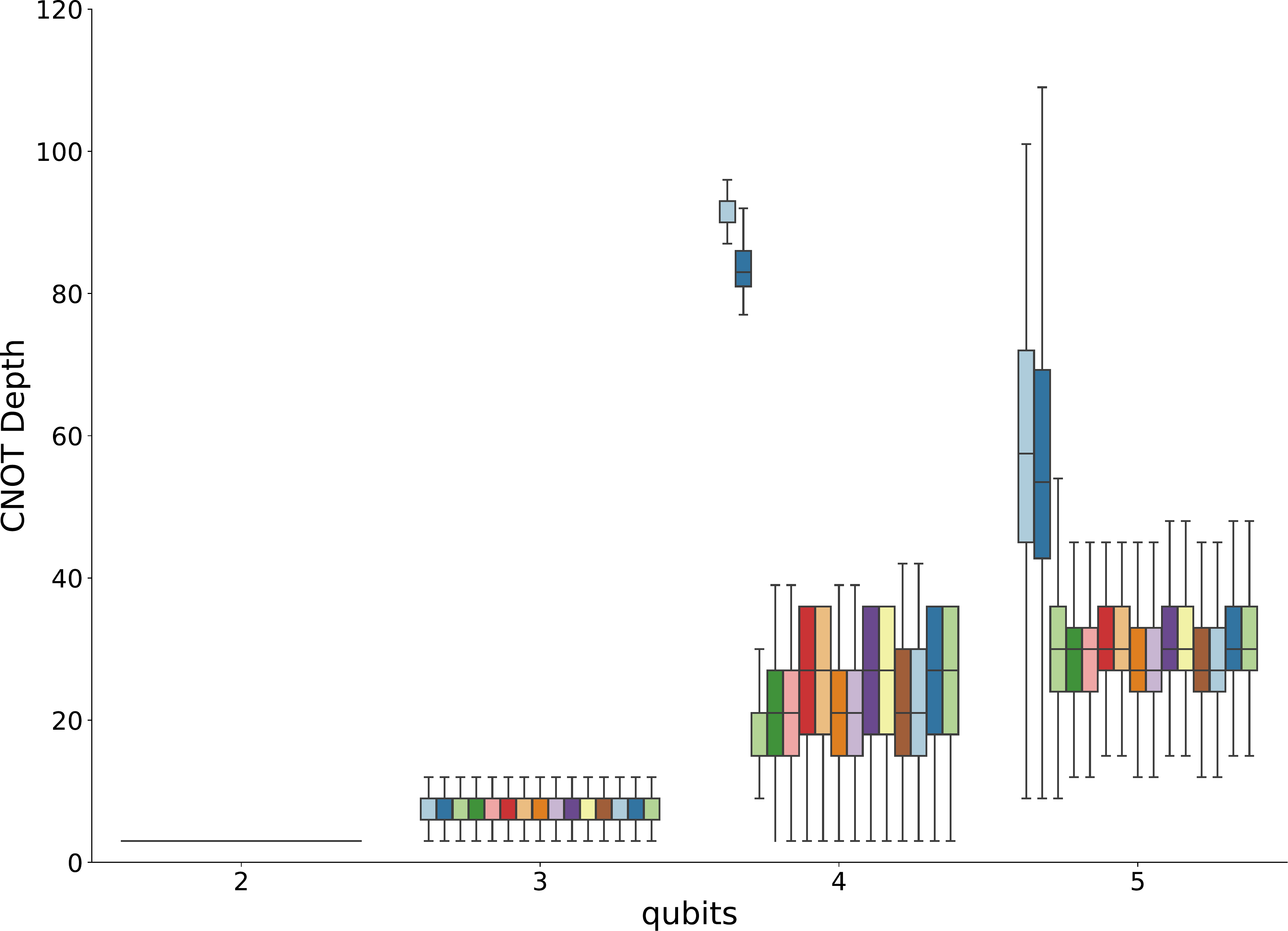}
    	\subcaption{}
        \label{fig:square_cx_depth}
    \end{minipage}
    &
    \begin{minipage}{7cm}
		\centering
		\includegraphics[width=3cm]{legend.pdf}
	\end{minipage}
	\end{tabular}
    \caption{Comparison of compilation strategies with square circuits using (\ref{fig:square_fidelity})~Hellinger fidelity, (\ref{fig:square_hog})~HOG, (\ref{fig:square_l1})~$l_1$-norm, (\ref{fig:square_cx_count})~CNOT count and (\ref{fig:square_cx_depth})~CNOT depth as metrics.  Qiskit’s statevector simulator has been used to sample noisy circuits from the noise model obtained with calibration data of the IBM \textit{ibmq\_casablanca} 7-qubit device.}
    \label{fig:square}
\end{figure*}

\begin{figure*}[!ht]
\centering
\begin{tabular}{ cc }
    \begin{minipage}{0.45\linewidth}
    	\centering
    	\includegraphics[width=\linewidth]{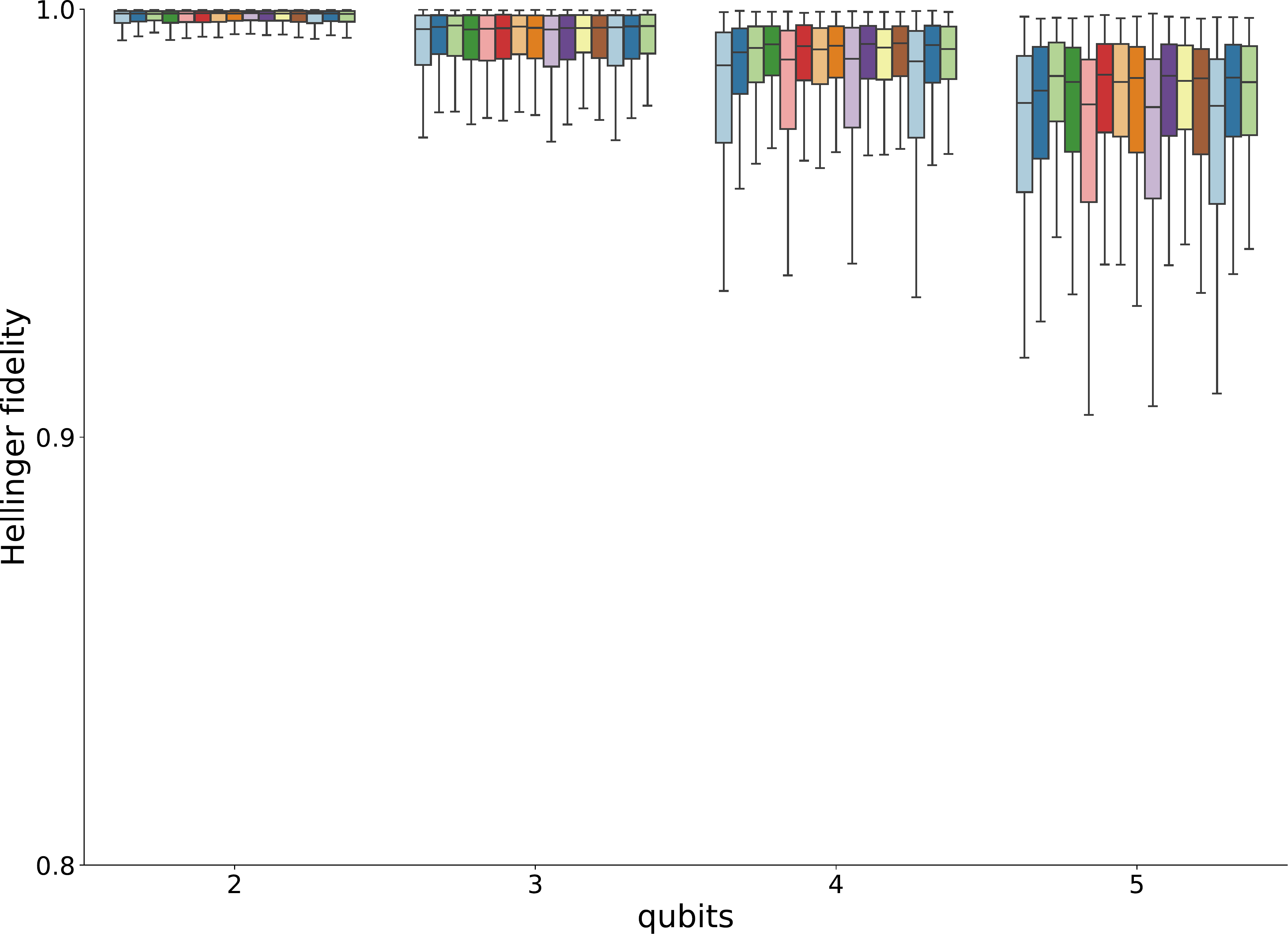}
    	\subcaption{}
        \label{fig:shallow_fidelity}
    \end{minipage}
    &
    \begin{minipage}{0.45\linewidth}
    	\centering
    	\includegraphics[width=\linewidth]{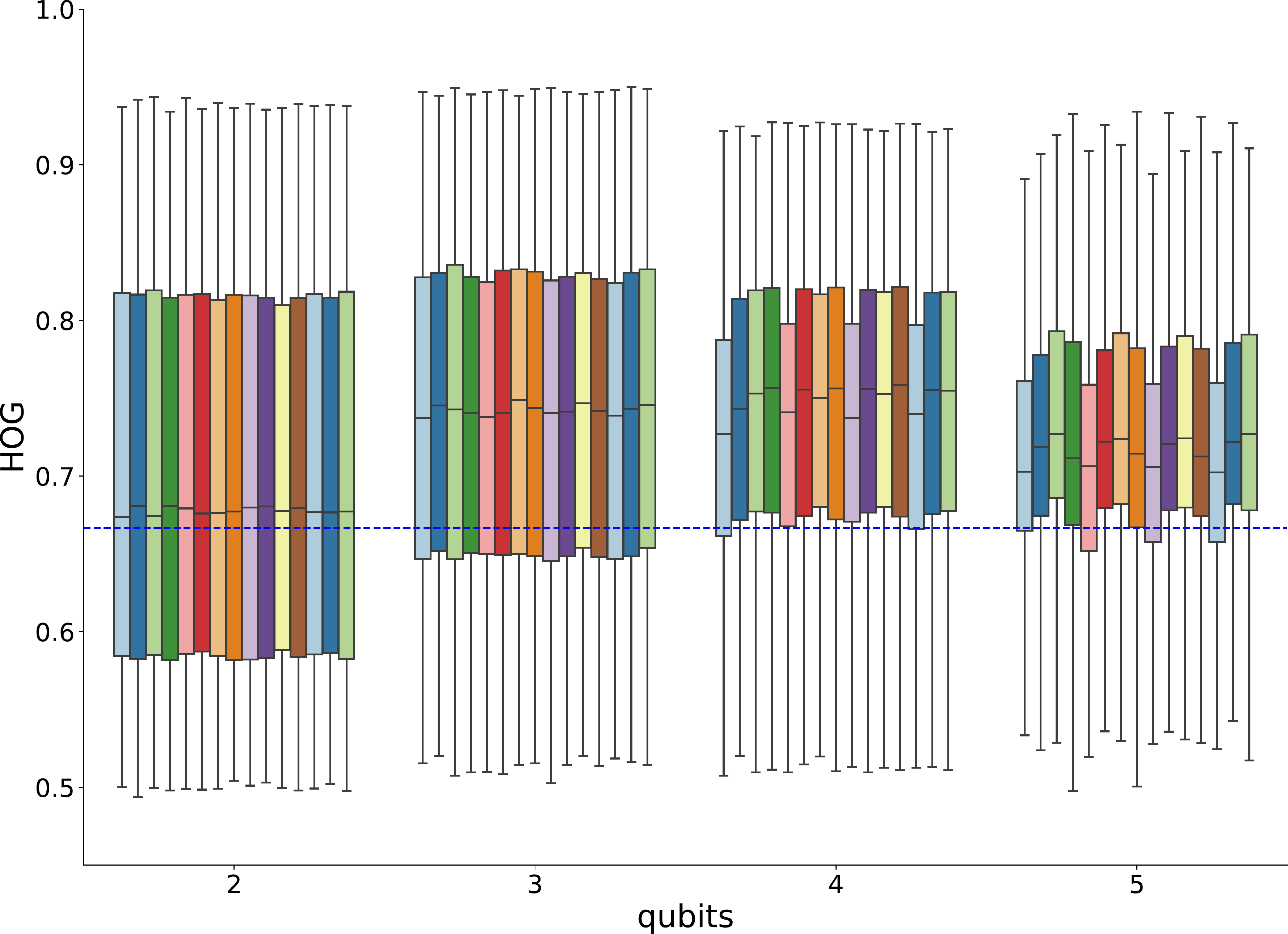}
    	\subcaption{}
        \label{fig:shallow_hog}
    \end{minipage}\\
    \begin{minipage}{0.45\linewidth}
		\centering
		\includegraphics[width=\linewidth]{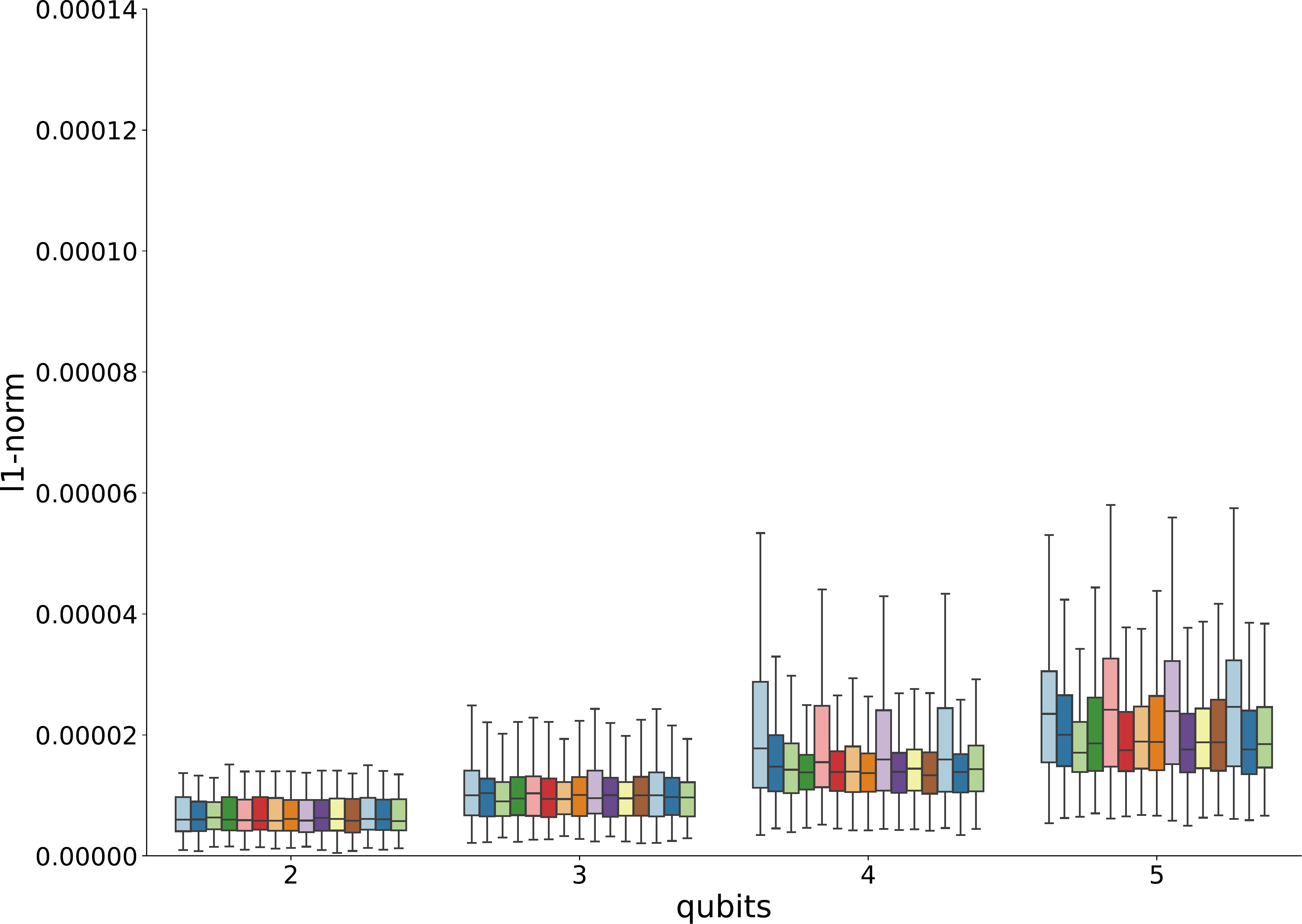}
		\subcaption{}
		\label{fig:shallow_l1}
	\end{minipage}
	&
	\begin{minipage}{0.45\linewidth}
    	\centering
    	\includegraphics[width=\linewidth]{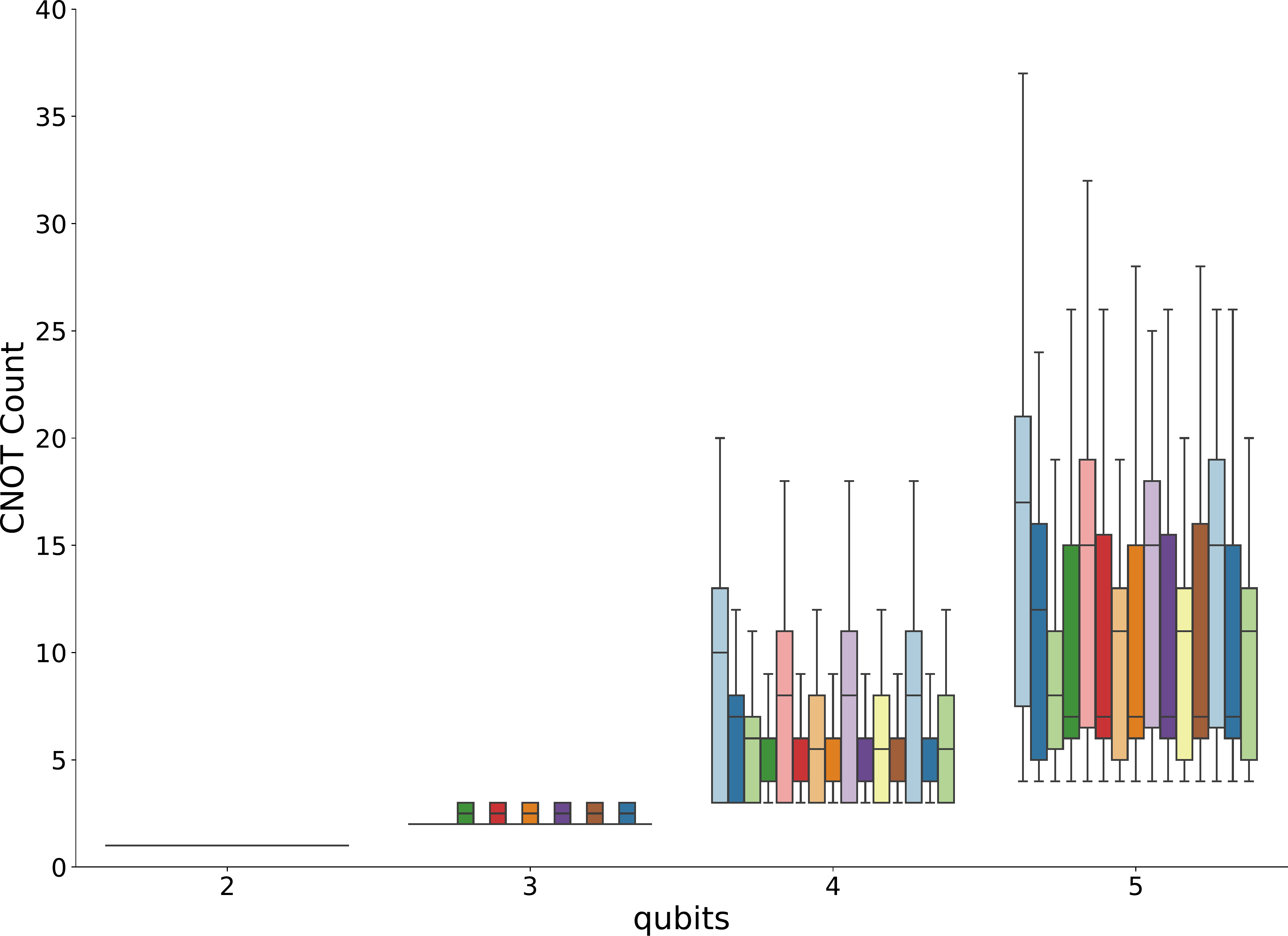}
    	\subcaption{}
        \label{fig:shallow_cx_count}
    \end{minipage}\\
    \begin{minipage}{0.45\linewidth}
    	\centering
    	\includegraphics[width=\linewidth]{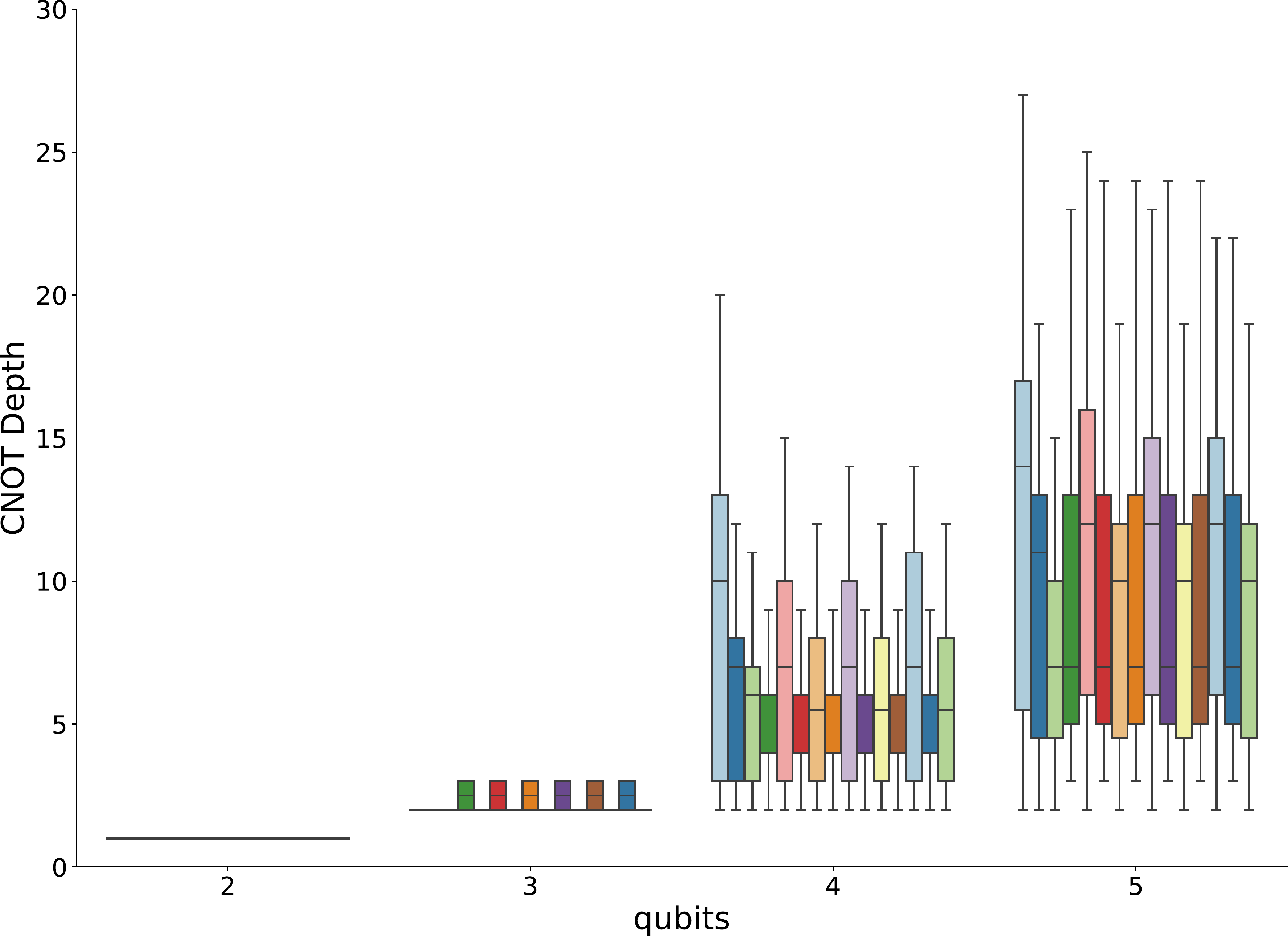}
    	\subcaption{}
        \label{fig:shallow_cx_depth}
    \end{minipage}
    &
    \begin{minipage}{7cm}
		\centering
		\includegraphics[width=3cm]{legend.pdf}
	\end{minipage}
	\end{tabular}
    \caption{Comparison of compilation strategies with shallow circuits using (\ref{fig:shallow_fidelity})~Hellinger fidelity, (\ref{fig:shallow_hog})~HOG, (\ref{fig:shallow_l1})~$l_1$-norm, (\ref{fig:shallow_cx_count})~CNOT count and (\ref{fig:shallow_cx_depth})~CNOT depth as metrics.  Qiskit’s statevector simulator has been used to sample noisy circuits from the noise model obtained with calibration data of the IBM \textit{ibmq\_casablanca} 7-qubit device.}
    \label{fig:shallow}
\end{figure*}

We have evaluated the proposed noise-adaptive compilation strategy, compared with some of the most advanced state-of-art approaches, over the circuit classes presented in Section~\ref{sec:benchmarks}, with 200 circuits for each class and different number of qubits, as in~\cite{Mills2020}. For each circuit the ideal output distribution $p_C$ has been obtained using Qiskit's statevector simulator. The same simulator has been used to sample noisy circuits from the noise model obtained with calibration data of the IBM \textit{ibmq\_casablanca} 7-qubit device.

We have tested the following compilers:
\begin{itemize}
    \item \textit{qiskit} is the Qiskit standard compiler with optimization level 3;
    \item \textit{qiskit\_noise} is the Qiskit compiler with \textsf{NoiseAdaptiveLayout} and optimization level 3;
    \item \textit{pytket} is the t$|$ket$\rangle$ compiler with \textsf{NoiseAwarePlacement} and maximum optimization level;
    \item \textit{na} are the proposed noise-adaptive placement and routing passes integrated with Qiskit and optimization level 3;
    \item \textit{qiskit\_na} is the proposed routing pass combined with Qiskit's \textsf{NoiseAdaptiveLayout} and optimization level 3;
    \item \textit{t} denotes the application of CNOT cascade transformations, as described in~\cite{Ferrari2021}.
\end{itemize}
For each of the above we have used Qiskit v0.23.6 and t$|$ket$\rangle$ v0.7.2. We remark that the \textit{pytket\_na} compiler was not tested, as the integration of \textit{pytket} and \textit{na} is possible only with regards to swap passes, not for the placement ones.

The results are reported in Fig.\ref{fig:deep}, Fig.\ref{fig:square} and Fig.\ref{fig:shallow}.
For each compiler configuration, we have used box plots to depict the distributions of the resulting values for the figures of merit described in Section~\ref{sec:benchmarks}. A box plot is constructed of two parts, a box and a set of whiskers. The box is drawn from the lower quartile to the upper quartile with a horizontal line drawn in the middle to denote the median. For the whiskers, the lowest point is the minimum of the data set and the highest point is the maximum of the data set.

The reader may observe that, for deep and square circuits, the best results in Hellinger fidelity, HOG, and $l_1$-norm are those achieved with \textit{qiskit\_na}. Interestingly, with this compiler, the CNOT depth is worst than with the other compilers. Indeed, improving the quality of a quantum computation with a noise-adaptive compilation strategy does not forcedly imply reducing the depth of the circuit. Regarding shallow circuits, \textit{qiskit\_na} and \textit{pytket} have the same performance.

Regarding the impact of $\alpha$, one may observe that, with deep circuits, $\alpha=0.9$ produces better Hellinger fidelity. Instead, with square and shallow circuits, $\alpha=0.5$ seems a better choice. The difference is nevertheless minimal. This could be due to the small number of qubits in the circuits and in the considered device.

\section{Conclusions}
\label{sec:conclusion}

In this work, we presented a novel noise-adaptive quantum compilation strategy that is computationally efficient.
The contributed strategy assumes heavy-hexagon topologies for quantum devices, which is particularly crucial for the placement pass. Moreover, the assumption simplifies the derivation of the computational complexity upper bounds. Nevertheless, the proposed routing pass is general enough to be effective independently of the coupling map of the target quantum device.

The presented results seem to indicate that our compilation strategy is particularly effective for circuits characterized by great depth and/or randomness. On the other hand, we are aware that the performed evaluation is not exhaustive and further work is necessary to fully characterize the proposed approach, for example in designing a reasonable method for tuning the $\alpha$ parameter in Eq.~\ref{eq:h}. In particular, we plan to extend the evaluation to circuits and devices with more qubits, and to distributed quantum computing architectures as well \cite{VanMeter2016,Cuomo2020,Ferrari2021bis}.

As a final remark, it is worth noting that it is possible to further mitigate the effect of noise by using an \textit{ensemble of diverse mappings} (EDM) approach, as suggested by Tannu and Qureshi \cite{Tannu2019}. In the near future, we shall integrate this method into our quantum compiling library.

\section*{Acknowledgements}

This research benefited from the HPC (High Performance Computing) facility of the University of Parma, Italy.

\bibliographystyle{unsrt}
\bibliography{main.bib}

\appendix

\section{Benchmark Circuits}

As stated in Section~\ref{sec:benchmarks}, the compilation strategies have been evaluated on the circuits proposed by Mills at al.~\cite{Mills2020}\footnote{\url{https://doi.org/10.5281/zenodo.3832121}}. The circuits are divided into three classes. Figure~\ref{fig:gates} shows the distributions of the number of gates in the considered circuits, for each class and number of involved qubits. 

\begin{figure}[!ht]
\centering
\begin{tabular}{c}
\begin{minipage}{0.58\linewidth}
    \centering
    \subcaption{Deep circuits.}
    \includegraphics[width=\linewidth]{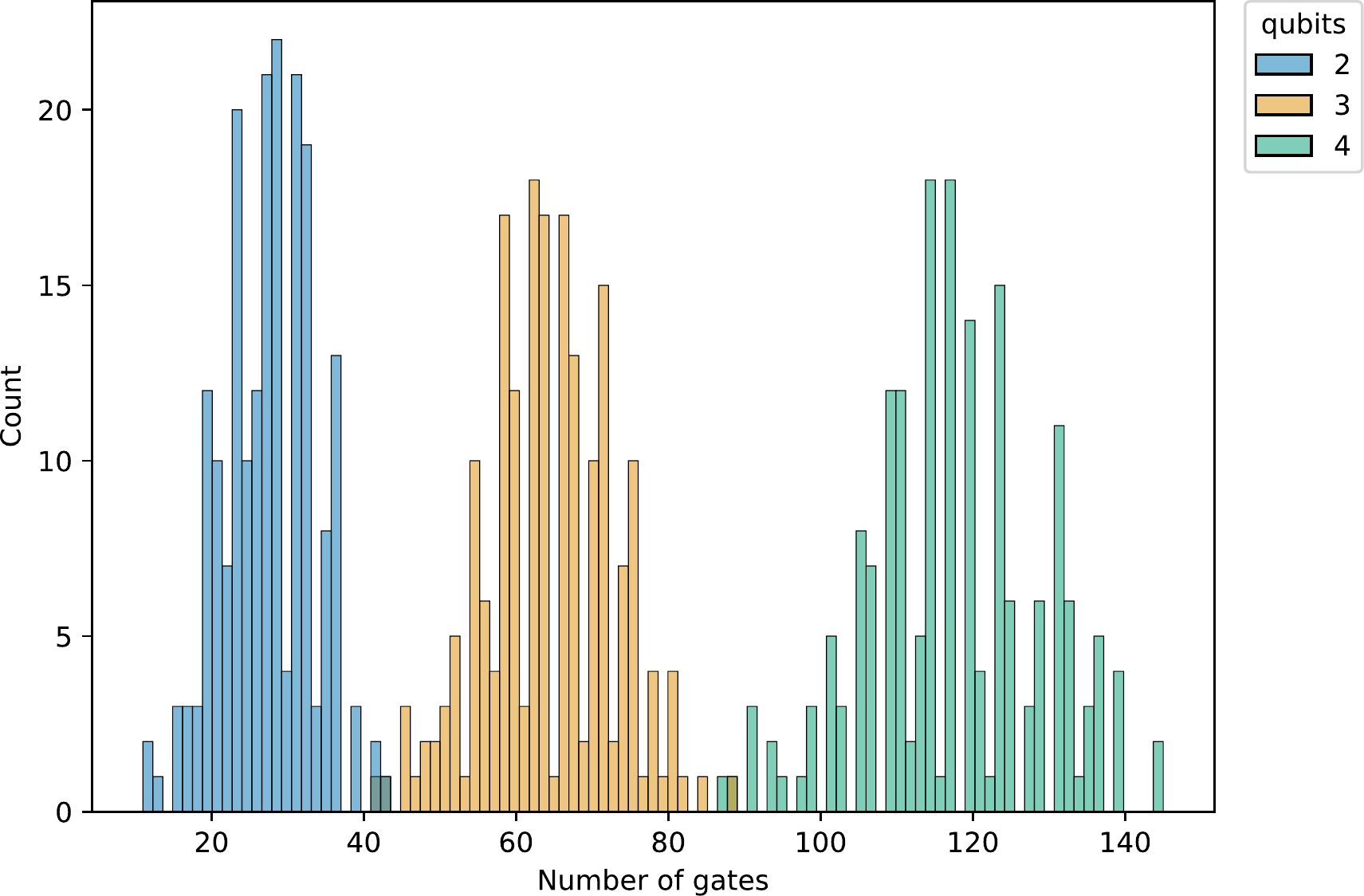}
    \label{fig:deep_gates}
\end{minipage}
\\    
\begin{minipage}{0.58\linewidth}
    \centering
    \subcaption{Square circuits}
    \includegraphics[width=\linewidth]{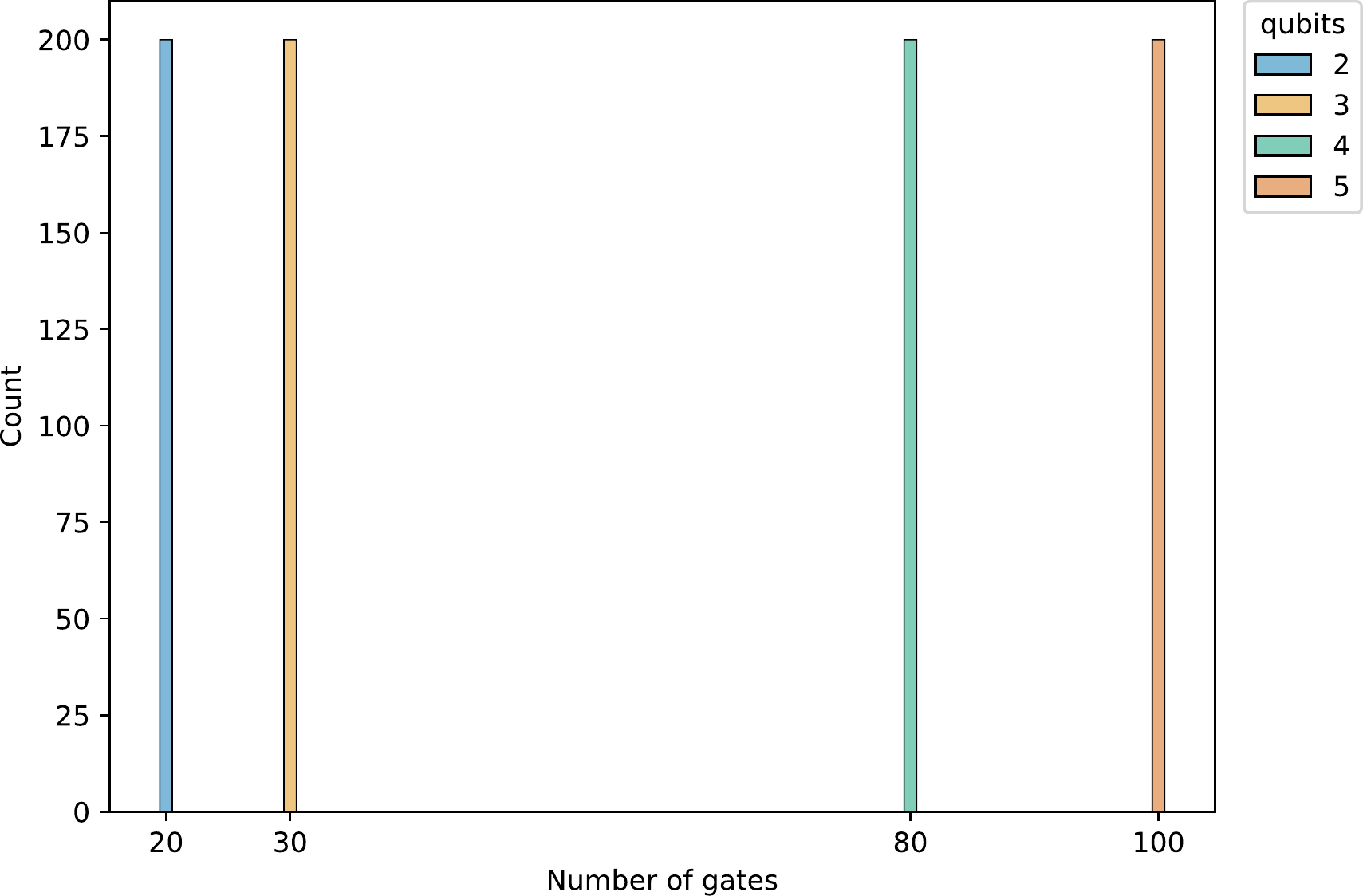}
    \label{fig:square_gates}
\end{minipage}
\\
\begin{minipage}{0.58\linewidth}
    \centering
    \subcaption{Shallow circuits}
    \includegraphics[width=\linewidth]{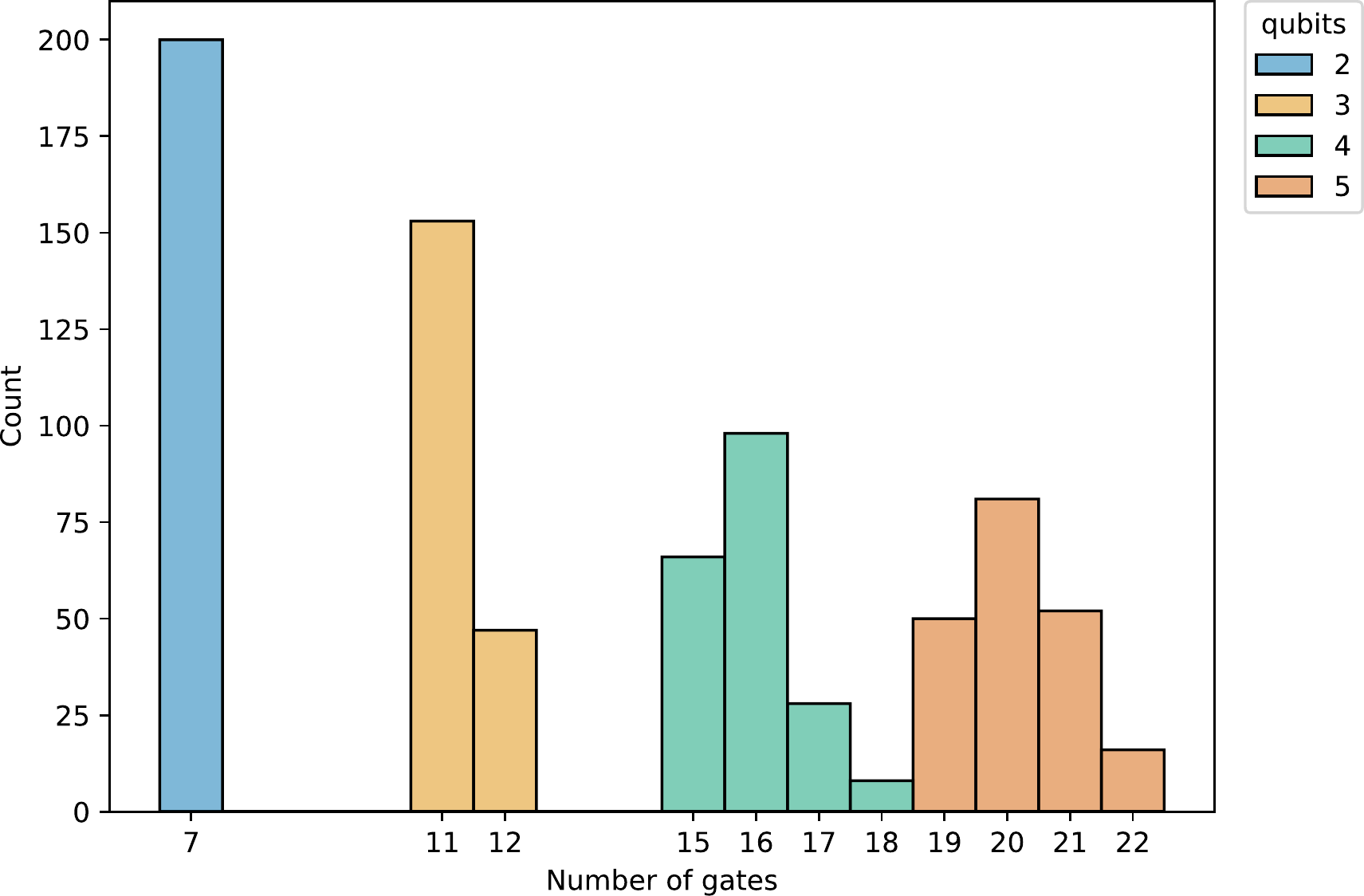}
    \label{fig:shallow_gates}
\end{minipage}
\end{tabular}
\caption{Distributions of the number of gates for each class of circuit, depending on the number of involved qubits. By construction, all the square circuits have the same number of gates, which increases with the number of involved qubits.}
\label{fig:gates}
\end{figure}

\end{document}